%% file: LagrangianJournal.tex
\documentclass[11pt,letterpaper]{article}

\pdfoutput=1
\usepackage[margin=2cm]{geometry}

\usepackage[utf8]{inputenc}
\usepackage{amsmath}
\usepackage{mathtools}                      % bigtimes
\usepackage{amssymb}
\usepackage{optidef}
\usepackage{mathrsfs}
\usepackage[ruled]{algorithm2e}
\usepackage{color}
\usepackage{colortbl}

\usepackage{amsthm}

\usepackage{graphicx}
\graphicspath{{./figures/}}

\input{macrosv1.tex}    % For bold symbols

\DeclareMathOperator{\boundeddist}{\mathcal{E}}
\DeclareMathOperator{\boundeddistover}{\mathcal{O}}

\DeclareMathOperator{\levelset}{\mathcal{L}}
\DeclareMathOperator{\prob}{\mathbb{P}}
\DeclareMathOperator{\minkdiff}{\ominus} % or \circleddash
\DeclareMathOperator{\minkadd}{\oplus}

\DeclareMathOperator{\targettube}{\mathscr{T}}
\DeclareMathOperator{\naturalnums}{\mathbb{N}}
\DeclareMathOperator{\realnums}{\mathbb{R}}
\DeclareMathOperator{\integers}{\mathbb{N}}

\def\dtinterval#1#2{\naturalnums_{[#1,#2]}}

\DeclareMathOperator{\targetset}{\mathcal{T}}
\DeclareMathOperator{\safeset}{\mathcal{K}}
\DeclareMathOperator{\rad}{Reach^{\flat}}
\DeclareMathOperator{\radover}{Reach\raisebox{2pt}{$^{\sharp}$}}
\DeclareMathOperator{\reach}{Reach}

\DeclareMathOperator{\statespace}{\mathcal{X}}  
\DeclareMathOperator{\inputspace}{\mathcal{U}}
\DeclareMathOperator{\distspace}{\mathcal{W}}

\def\onestepforwardreach#1{\mathscr{F}_{1, #1}}
\def\onestepbackwardreach#1{\mathscr{R}_{1, #1}}
\def\reachunder#1{\reach_{#1}^{\flat}(\targettube_{N}, \boundeddist)}
\def\multreachunder#1#2{\reach_{#1}^{\flat}(\targettube_{N}, \boundeddist_{#2})}
\def\multreachover#1#2{\reach_{#1}^{\sharp}(\targettube_{N}, \boundeddistover_{#2})}
\def\reachover#1{\reach_{#1}^{\sharp}(\targettube_{N}, \boundeddistover)}

\def\transpose#1{#1^{\top}}
\def\ONE{\boldsymbol{1}}

\def\lsc{l.s.c.}
\def\radnl#1#2{\reach_{#1}^{\flat+}(\targettube_{N}, \boundeddist_{#2})}
\def\radovernl#1#2{\reach_{#1}^{\sharp+}(\targettube_{N}, \boundeddistover_{#2})}
\def\radlin#1#2{\reach_{#1}^{\flat+}(\targettube_{N}, \boundeddist_{#2})}
\def\probdist#1{
    \let\argone=#1
    \ifnum \the\catcode`#1 = 11
        \prob_{\bar{w}_{[#1, N-1]}}
    \else
        \ifcat \argone 0
            \newcount\tempcount
            \tempcount=#1
            \prob_{\bar{w}_{[\the\tempcount, N-1]}}
        \else
            ???
        \fi
    \fi
}
\def\probxpi#1{
    \let\argone=#1
    \ifnum \the\catcode`#1 = 11
        \prob_{\bar{x}_{[#1+1, N]}}^{x_{#1}, \pi}
    \else
        \ifcat \argone 0
            \newcount\tempcount
            \tempcount=#1
            \prob_{\bar{x}_{[\advance\tempcount by 1 \the\tempcount, N]}}^{x_{\the\tempcount}, \pi}
        \else
            ???
        \fi
    \fi
}

\def\probxpistar#1{
    \let\argone=#1
    \ifnum \the\catcode`#1 = 11
        \prob_{\bar{x}_{[#1+1, N]}}^{x_{#1}, \pi^{\ast}}
    \else
        \ifcat \argone 0
            \newcount\tempcount=#1
            \prob_{\bar{x}_{[\advance\tempcount by 1 \the\tempcount, N]}}^{x_{\the\tempcount}, \pi^{\ast}}
        \else
            ???
        \fi
    \fi
}

\def\stochlevelset#1{\mathcal{L}_{#1}(\targettube_{N}, \alpha)}
\newtheorem{assumption}{Assumption}
\newtheorem{lemma}{Lemma}
\newtheorem{remark}{Remark}
\newtheorem{proposition}{Proposition}
\newtheorem{theorem}{Theorem}
\newtheorem{corollary}{Corollary}
\newtheorem{problem}{Problem}

\newtheorem{innercustomthm}{Problem}
\newenvironment{customthm}[1]
  {\renewcommand\theinnercustomthm{#1}\innercustomthm}
  {\endinnercustomthm}

\begin{document}
\title{Lagrangian Approximations for Stochastic Reachability of a Target Tube}
\author{Joseph D. Gleason$^\dagger$, Abraham P. Vinod$^\dagger$, and Meeko M. K. Oishi
    \thanks{This material is based upon work supported by the National Science
    Foundation under Grant Number IIS-1528047, CMMI-1254990 (Oishi, CAREER),
    CNS-1329878, and the Lighting Enabled Systems and Applications Engineering
    Research Center (EEC-0812056). Any opinions, findings, and conclusions or
    recommendations expressed in this material are those of the authors and do
    not necessarily reflect the views of the National Science Foundation.
    \newline \indent 
    Joseph Gleason, Abraham Vinod, and Meeko Oishi are with Electrical and
    Computer Engineering, University of New Mexico, Albuquerque, NM; e-mail:
    gleasonj@unm.edu, aby.vinod@gmail.com, oishi@unm.edu (corresponding author).
    \newline \indent
    % \thanks{
    $\dagger$ These authors contributed equally to this work.}
}
\date{}
\maketitle

\begin{abstract}
    In this paper we examine how Lagrangian techniques can be used to compute
    underapproximations and overapproximation of the finite-time horizon,
    stochastic reach-avoid level sets for discrete-time, nonlinear systems.
    This approach is applicable for a generic nonlinear system without
    any convexity assumptions on the safe and target sets.
    % Under assumptions of linearity and convexity, these
    % approximations can be efficiently computed using existing computational
    % geometry toolboxes.
    We examine and apply our methods on the reachability of a target tube
    problem, a more generalized version of the finite-time horizon reach-avoid
    problem.
    % We apply our methods on reachability of a target tube and show how this
    % formulation can include the standard terminal-time reach-avoid problem as
    % well as the viability problem, while also allowing for more complicated
    % verification problem analysis. 
    Because these methods utilize a Lagrangian (set theoretic) approach, we
    eliminate the necessity to grid the state, input, and disturbance spaces
    allowing for increased scalability and faster computation. The methods
    scalability are currently limited by the computational requirements for
    performing the necessary set operations by current computational geometry
    tools. The primary trade-off for this improved extensibility is conservative
    approximations of actual stochastic reach set.   
    % These methods, while conservative, do not rely on
    % gridding, allowing for scalability as permitted by the constraints present
    % in the computational geometry. 
    We demonstrate these methods on several examples including the standard
    double-integrator, a chain of integrators, and a 4-dimensional space vehicle rendezvous docking
    problem.
\end{abstract}

\section{Introduction}

Reach-avoid analysis is an established verification tool that provides formal
guarantees of both safety (by avoiding unsafe regions) and performance (by
reaching a target set). Because of these guarantees, it is often used in
systems that are safety-critical or expensive, such as space systems
\cite{lesser2013_spacecraft}, avionics \cite{Tomlin2003,summers2011stochastic},
biomedical systems \cite{maidens_2013}, and other applications
\cite{KariotoglouECC2011,ManganiniCYB2015,summers2010_verification}. The
reach-avoid set is the set of states for which there exists a control that
enables the state trajectory to reach a target at some finite time horizon,
$N$, while remaining within a safe set (avoiding an unsafe set) for all
instants in the time horizon. In a probabilistic system, satisfaction of the
reach-avoid objective is accomplished stochastically. A dynamic
programming-based solution characterizes the optimal value function, a function
that assigns to each initial state the optimal probability of achieving the
reach-avoid objective \cite{abate2008_reachability}. An appropriate level set
of this value function provides the {\em stochastic reach-avoid level set}, the
set of states for which probabilistic success of the reach-avoid objective is
assured with at least the desired likelihood.

The theoretical framework for the probabilistic reach-avoid calculation uses
dynamic programming \cite{summers2010_verification,abate2008_reachability},
and, hence, is computationally infeasible for even moderate-sized systems due
to the gridding of not only the state-space, but also of the input and
disturbance spaces \cite{AbateHSCC2007}. Alternatives to dynamic programming
include approximate dynamic programming \cite{ManganiniCYB2015,
KariotoglouECC2013, KariotoglouSCL2016}, Gaussian mixtures
\cite{KariotoglouSCL2016}, particle filters \cite{lesser2013_spacecraft,
ManganiniCYB2015}, and convex chance-constrained optimization
\cite{lesser2013_spacecraft, KariotoglouECC2011}. These methods have been
applied to systems that are at most 10-dimensional, at high memory and
computational costs \cite{ManganiniCYB2015}. Further, since an analytical
expression of the value function is not accessible, stochastic reach-avoid
level sets can be computed only up to the accuracy of the gridding. Recently, a
Fourier transform-based approach has provided greater scalability verifying LTI
systems of dimension up to $40$ \cite{VinodLCSS2017, VinodHSCC2018}.
In~\cite{VinodHSCC2018} the researchers established a set of sufficient
conditions in which the stochastic reach-avoid set is convex and compact,
enabling scalable polytopic underapproximation. However, this approach relies
on numerical quadrature and is restricted to verifying the existence of
open-loop controllers.

For deterministic systems (that is, systems without a disturbance input but
with a control input) or systems with disturbances that come from a bounded
set, Lagrangian methods for computing reachable sets are popular because they
do not rely on a grid and can be computed efficiently for high-dimensional
systems \cite{maidens_2013, guernic2009,borelli_2017}. Rather than gridding the
system, Lagrangian methods compute reachable sets through operations on sets,
e.g. intersections, Minkowski summation, unions, etc. Thus, Lagrangian methods
rely on computational geometry, whose scalability depends on the set
representation and the computational difficulty of the operations used
\cite{guernic2009}. For example, sets that are represented as either vertex or
facet polyhedra typically are limited by the need to solve the {\itshape
vertex-facet enumeration problem}. Common set representations and relevant
toolboxes for their implementation are: polyhedrons (\texttt{MPT} \cite{MPT3}),
ellipsoids (\texttt{ET} \cite{ellipsoidal}), zonotopes \cite{girard2005}
(\texttt{CORA} \cite{althoff2015}), star representation (\texttt{HyLAA}
\cite{bak2017}), and support functions \cite{leguernic2010}.

In this paper, we describe recursive techniques to obtain an under and an
overapproximation of the probabilistic reach-avoid level set using Lagrangian
methods. 
% For the underapproximation, we first reiterate and extend key results
% from \cite{GleasonCDC2017} which details theory for using Lagrangian methods
% for computing conservative underapproximations. The underapproximation
% represents a set of states which we can guarantee has a probability of
% achieving the reach-avoid objective with a probability value of at least
% $\alpha \in [0,1]$. 
The underapproximation can be theoretically posed as the solution to the
\emph{reachability of a target tube} problem
\cite{bertsekas1971minimax,kerrigan2001robust,rakovic2006_reach}, originally
framed to compute reachable sets of discrete-time controlled systems with
bounded disturbance sets.
Motivated by the scalability of the Lagrangian method
proposed in~\cite{maidens_2013, saintpierre1994_viability} for viability
analysis in deterministic systems (that is, systems without a disturbance input
but with a control input), we seek a similar approach to compute the underapproximation 
via tractable set theoretic operations. We originally demonstrated these methods
in \cite{GleasonCDC2017} unifying these approaches for the terminal-time 
reach-avoid problem. The reachability of a target tube problem is, however, a more
generalized framework than the terminal-time reach-avoid problem.
% by computing the {\em robust
% reach-avoid set}, the set of states assured to reach the target set and remain
% in the safe region despite {\em any} disturbance input. 
Hence, we extend \cite{GleasonCDC2017} to the reachability of a target tube 
problem and additionally describe a recursive method for computing an overapproximation
to the stochastic reach-avoid level set. Borrowing notation and terminology from
\cite{kaynama2011} we call these under and overapproximations \emph{disturbance
minimal} and \emph{disturbance maximal reach sets}, respectively.
%  application with target tubes, thus,
% will henceforth refer to the \emph{robust reach-avoid set} as the \emph{robust
% effective target set} for reasons that will be discussed in the subsequent.

The disturbance minimal reach set (underapproximation) is the set of initial
states of the system for which there exists a disturbance that will remain
in the target tube despite the worst case choice of a disturbance drawn from 
a bounded set $\boundeddist$, i.e. for all disturbances in $\boundeddist$. The
disturbance maximal reach set (overapproximation) is the set of states for which there exists and
input that will remain in the target tube given the best choice of a disturbance
in a bounded set $\boundeddistover$, i.e. exists a disturbance in $\boundeddistover$.
The original formulation \cite{GleasonCDC2017} described the underapproximation
for a single bounded disturbance set $\boundeddist$. Here we will also demonstrate
that we can use many different bounded disturbance sets $\boundeddist_{i}$, 
$\boundeddistover_{i}$, $i \in \{1, \dots, M\}$, to help reduce the conservativeness
of the approximations. Using multiple bounded disturbance sets requires repeated
computations of the disturbance minimal and maximal reach sets. However, because
of the reduced computational time from using Lagrangian methods we can use
multiple disturbance sets and still provide substantially faster computation.

    % In the original formulation \cite{GleasonCDC2017} we determined an underapproximation
    % using a single bounded disturbance set. Here we will additionally detail how
    % using many different sets can improve the approximation, i.e. reduce the 
    % conservativeness, 

% For the overapproximation, we utilize similar methods as in
% \cite{GleasonCDC2017,gleason_2018} and describe theory for computing an
% overapproximation of the reach-avoid level set. This overapproximation provides
% a set of states for which if $x$ does not lie in this set, the probability of
% achieving the reach-avoid objective cannot exceed $\alpha$. This approximation
% is also applied to target tubes and is achieved by computing what we will call
% the {\em augmented effective target set}.

The remainder of the paper is as follows: Section \ref{sec:preliminaries}
provides the necessary preliminaries and describes the problem formulation. In
Section \ref{sec:stochastic-approx}, we establish sufficient conditions for the
bounded disturbance sets such that the disturbance minimal
and maximal reach sets are a subset and superset of the stochastic reach set,
respectively. Section 
\ref{sec:lagrangian-methods-recursion}  
details the recursions for computing the disturbance minimal and disturbance maximal reach sets,
given a single sufficient bounded disturbance set.
Section \ref{sec:algorithm_and_challenges} describes how
these approximations can be improved by using multiple bounded disturbance
sets. We examine numerical methods to obtain these sufficient bounded disturbance
sets in Section \ref{sec:computation_of_boundeddist_boundeddistover} and examine
the numerical implementation challenges for computing the disturbance
minimal and maximal reach sets in Section \ref{sec:computational_challenges}. 
Finally, we demonstrate our algorithm on
selected examples in Section \ref{sec:examples} and provide conclusions and
directions of future work in Section \ref{sec:conclusion}.

\section{Preliminaries}
\label{sec:preliminaries}

The following notation will be used throughout the paper. We denote the set of
natural numbers, including zero, as $\naturalnums$, and discrete-time intervals
with $\dtinterval{a}{b} = \naturalnums \cap \{ a, a+1, \dots, b-1, b\}$, for
$a, b \in \naturalnums$, $a \leq b$. We will primarily use $k \in \naturalnums$
as our discrete-time index but will also use $t \in \naturalnums$ as necessary
to provide clarity. The transposition of a vector is $\transpose{x}$, and the
concatenation of a discrete-time series of vectors is noted with a bar above
the variable and a subscript with the indices, i.e. $\bar{x}_{[k,N]} =
\transpose{[\transpose{x_{k}}, \transpose{x_{k+1}}, \dots,
\transpose{x_{N}}]}$, $x_{t} \in \realnums^{n}$ for $t \in \dtinterval{k}{N}$.
The $n$-dimensional identity matrix is noted as $I_{n}$.

The Minkowski summation of two sets $ \mathcal{S}_1, \mathcal{S}_2 \subseteq
\realnums^{n}$ is $ \mathcal{S}_1 \oplus \mathcal{S}_2=\{s_1+s_2:s_1\in
\mathcal{S}_1, s_2 \in \mathcal{S}_2\}$; the Minkowski difference (or
Pontryagin difference) of two sets $ \mathcal{S}_1, \mathcal{S}_2$ is $
\mathcal{S}_2 \minkdiff \mathcal{S}_1=\{s:s+s_1\in \mathcal{S}_2\ \forall s_1
\in \mathcal{S}_1\}$. For $\statespace\subseteq \realnums^n, n\in \mathbb{N},
n>0$, the indicator function corresponding to a set $ \mathcal{S}$ is
$\ONE_{\mathcal{S}}: \statespace \rightarrow \{0,1\}$ where
$\ONE_{\mathcal{S}}(x)=1$ if $x\in \mathcal{S}$ and is zero otherwise; the
Cartesian product of the set $\mathcal{S}$ with itself $k \in \mathbb{N}$ times
is $\mathcal{S}^{k}$.

\subsection{Lower semi-continuity}

Lower semi-continuous (l.s.c.) functions are functions whose sub-level sets are
closed~\cite[Definition 7.13]{BertsekasSOC1978}. Lower semi-continuous
functions have very useful properties with respect to optimization.

\begin{itemize}
    \item[(P1)] \emph{Indicator function of a closed set is \lsc{}}: Given a closed set $S\subseteq \mathcal{X}$, the function $-\ONE_{S}(x)$ is \lsc{}\footnote{We know that $\ONE_{S}(x)$ is upper semi-continuous~\cite[Definition 7.13]{BertsekasSOC1978}, and the negation of an upper semi-continuous function yields a \lsc{} function.}.
    \item[(P2)] \emph{Addition and \lsc{}}: Given two \lsc{} functions $l_1(x),l_2(x): \mathcal{X} \rightarrow \realnums$ such that $\forall x\in \mathcal{X}$, $\vert l_1(x)\vert \neq\infty$ and $\vert l_2(x)\vert \neq \infty$, the function $l_1(x)+l_2(x)$ is \lsc{} over $ \mathcal{X}$~\cite[Ex. 1.39]{rockafellar_variational_2009}. 
    \item[(P3)] \emph{Semicontinuity under composition}: Given a \lsc{} functions $l(x): \mathcal{X} \rightarrow \realnums$ and a continuous function $C(x): \mathcal{Y} \rightarrow \mathcal{X}$, the function $l\circ C=l(C(\cdot)): \mathcal{Y} \rightarrow \realnums$ is \lsc{} over $ \mathcal{Y}\subseteq \realnums^m,\ m\in \mathbb{N},\ m>0$~\cite[Ex. 1.40]{rockafellar_variational_2009}.
    \item[(P4)] \emph{Supremum of a \lsc{} function}: Given $l: \mathcal{X}\times \mathcal{Y} \rightarrow \realnums\cup\{-\infty,\infty\}$ such that $l$ is \lsc{}, then $l^\ast(x)=\sup_{y\in \mathcal{Y}}l(x,y)=-\inf_{y\in \mathcal{Y}}(-l(x,y))$ is \lsc{}~\cite[Prop. 7.32(b)]{BertsekasSOC1978}.
    \item[(P5)] \emph{Infimum of a \lsc{} function}: Given $l: \mathcal{X}\times \mathcal{Y} \rightarrow \realnums\cup\{-\infty,\infty\}$ such that $l$ is \lsc{} and $ \mathcal{Z}\subseteq \mathcal{Y}$ is compact, then $l^\ast(x)=\inf_{y\in \mathcal{Z}}l(x,y)$ is \lsc{}. 
        Additionally, there exists a Borel-measurable $\phi: \mathcal{X} \rightarrow \mathcal{Z}$ such that $l^\ast(x)=l(x,\phi(x)),\ \forall x\in \mathcal{X}$~\cite[Prop. 7.33]{BertsekasSOC1978}.
\end{itemize}

% subsection Real analysis (end)

\subsection{System description}

Consider a discrete-time, nonlinear, time-varying dynamical system with an
affine disturbance,
\begin{equation}
  x_{k+1} = f_{k}(x_{k}, u_{k}) + w_{k}
  \label{eq:nonlin}
\end{equation}
with state $x_{k} \in \statespace \subseteq \realnums^{n}$, input $u_{k} \in
\inputspace \subseteq \realnums^{m}$, disturbance $w_{k} \in
\distspace\subseteq \realnums^n$, and a function $f: \statespace\times
\inputspace \rightarrow \statespace$. We denote the origin of $ \realnums^{n}$
as $0_{n}$ and assume $0_{n} \in\distspace$ without loss of generality due to the
affine nature of the disturbance. We also consider the discrete-time LTV system 
\begin{equation}
    x_{k+1} = A_{k} x_{k} + B_{k} u_{k} + w_{k} \label{eq:lin}
\end{equation}
with $A_{k}\in \realnums^{n\times n}$ and $B_{k}\in \realnums^{n\times m}$. We assume
$A_{k}$ is non-singular, which holds true especially for discrete-time systems that
arise from sampling continuous-time systems. We will consider the cases where
$w_{k}$ is uncertain (non-stochastic disturbance drawn from a bounded set) and
stochastic (random vector drawn from a known probability density function). The
discrete horizon length for the reachability of a target tube problem is marked
as $N \in \naturalnums$, $N>0$.

\subsection{Reachability of a target tube}
\label{sub:reachability-of-a-target-tube}

As in \cite{bertsekas1971minimax, bertsekasDP}, we define a \emph{target tube}
$\mathscr{T}_{N} = [\targetset_{0}, \targetset_{1}, \dots, \targetset_{N}]$,
as an indexed collection of subsets of the state space,
$\targetset_{k} \subseteq \statespace$, for all $ k \in \dtinterval{0}{N}$. We
assign attributes of the tube that are typically given to sets, e.g. closed,
bounded, compact, convex, etc., if and only if every set in the target tube has
those properties. For example, we say that the target tube $\targettube_{N}$ is
closed if and only if $\targetset_{k}$ is closed for all $k \in
\dtinterval{0}{N}$.

\begin{figure*}
    \centering
    \includegraphics{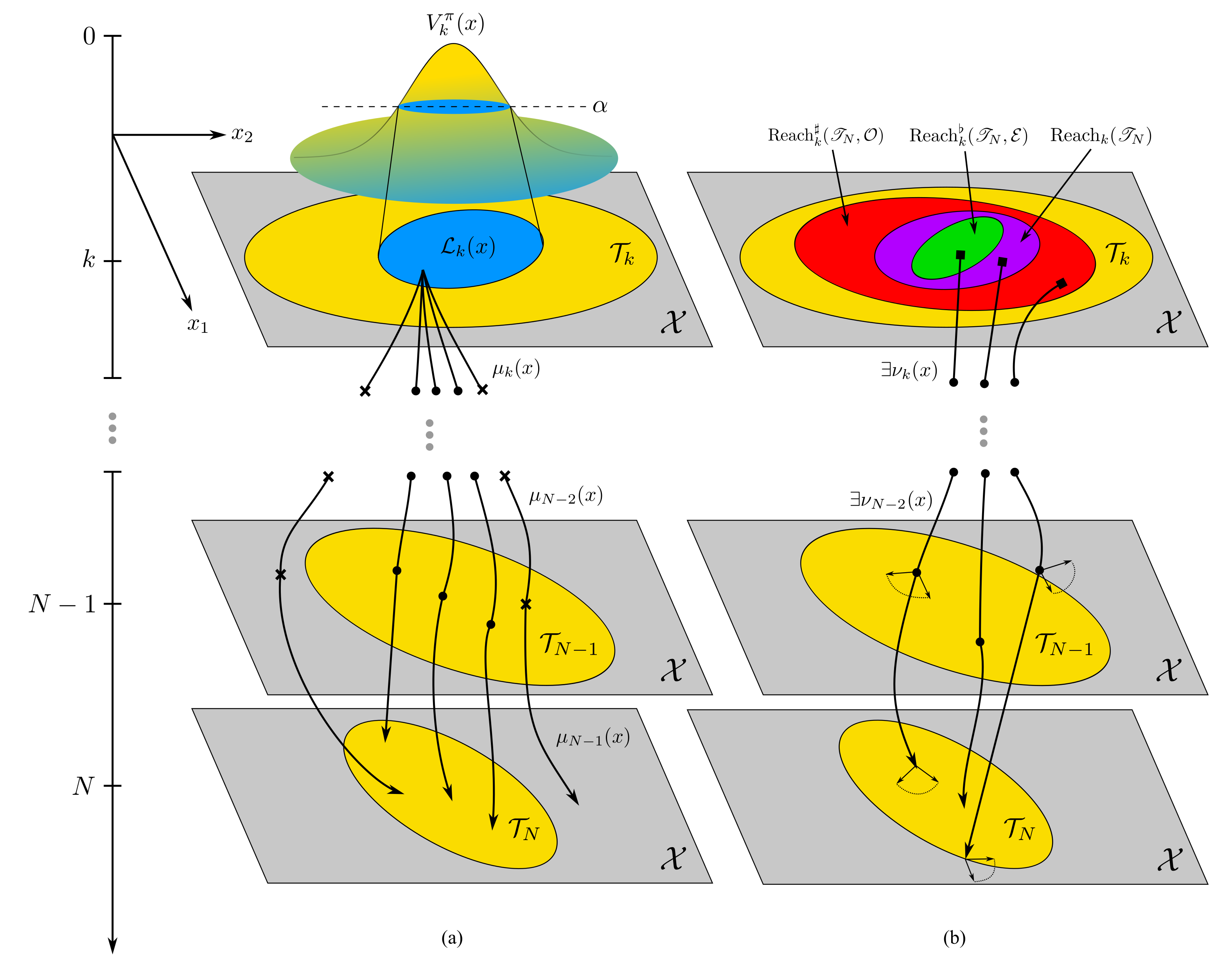}
    \caption{
        Graphical representation of (a) the stochastic $\alpha$-level reach
        set, $\stochlevelset{k}$ and (b) of the reachability of the target 
        tube---$\reach_{k}(\targettube_{N})$, \eqref{eq:reachability-of-target-tube}, 
        the disturbance minimal reach set---$\reachunder{k}$, 
        \eqref{eq:dist-minimal-reach-set}, and the disturbance maximal reach 
        set---$\reachover{k}$, \eqref{eq:dist-maximal-reach-set}. The 
        trajectories with circle markings (\textbullet) remain in the target 
        tube while trajectories with an ex ({\bfseries \texttt x}) fail to 
        remain in the tube. For the stochastic reach set starting from the same
        initial condition, some trajectories will not remain in the target 
        tube due to the random nature of the disturbance. For the disturbance
        minimal and maximal reach reach set the trajectories need to remain
        in the target set for all disturbances in $\boundeddist$, and some
        disturbance in $\boundeddistover$, respectively. The reachability of
        a target tube set, $\reach_{k}(\targettube_{N})$, does not consider a 
        disturbance and, hence, has a deterministic trajectory.
        % Example of target tube for aerial vehicle tracking ground vehicle driving on a curved road. The red circles represent the target sets at time instances $k \in \dtinterval{0}{N}$.
    }
    \label{fig:reach-sets-graphic}

    \makeatletter
    \@bsphack 
    \protected@write \@auxout {}{\string \newlabel {fig:stoch-reach-set}{{\@currentlabel (a)}{\thepage }}}
    \protected@write \@auxout {}{\string \newlabel {fig:uncertain-reach-sets}{{\@currentlabel (b)}{\thepage }}}
    \@esphack
    \makeatother
\end{figure*}

We will denote an admissible state-feedback law using $\nu: \statespace
\rightarrow \inputspace$, the set of admissible state-feedback laws with $
\mathcal{F}$, and the set of admissible control policies $\rho = [\nu_0(\cdot),
\nu_1(\cdot), \dots, \nu_{N-1}(\cdot)]$ using $\mathcal{P} = \mathcal{F}^N$.

The reachability of a target tube problem is concerned with determining the 
set of states at time $k \in \naturalnums$ for which there exists a control 
policy, $\rho$, such that, for the system \eqref{eq:nonlin} with no disturbance,
i.e. $\mathcal{W} = \{0_{n}\}$, the trajectory
$[x_{k}, \dots, x_{N}]$ will remain in the target tube. Formally, 
\begin{align}
    \reach_{k}(\targettube_{N}) = \big\{ x_{k} \in \targetset_{k} : \; &\forall t \in \dtinterval{k}{N-1}, \exists \nu_{t} \in \mathcal{F}, \nonumber\\
        &x_{t+1} \in \targetset_{t+1} \big\}
    \label{eq:reachability-of-target-tube}
\end{align}
This set can be seen graphically in Figure \ref{fig:uncertain-reach-sets}.

The motivation for using target tubes is to allow for more versatility in the
problem definition. We can unite the standard terminal-time reach-avoid problem---in which we
desire to reach a target, $\targetset$, at time $N$ while remaining in a safe
set, $\safeset$, for $k < N$---and the reachability of target tubes with
$\targettube_{N}^{RA} = [\safeset, \safeset, \dots, \safeset, \targetset]$. The
viability problem can equivalently be subsumed with $\targettube_{N}^{V} =
[\safeset, \safeset, \dots, \safeset]$. However many interesting problems fit
outside of the standard terminal-time reach-avoid or viability problems.

% Consider, for example, a problem in which we would like to have an aerial
% vehicle track a moving ground vehicle. We will presume, for simplicity of this
% example, is driving along a curved road. We will consider the aerial vehicle to
% be tracking the ground vehicle if it exists in the airspace above the ground
% vehicle within a ball of radius, $r$, whose center is the center of the ground
% vehicle. In this example it is clear that, since the ground vehicle is moving
% with a given trajectory, the target states would also be moving with a given
% trajectory. Figure \ref{fig:reach-sets-graphic} shows a graphical
% representation of this problem.

\subsection{Stochastic reachability of a target tube}

The basic reachability of a target tube problem
\eqref{eq:reachability-of-target-tube} is described with no disturbance. Here
we consider \eqref{eq:nonlin}, with a stochastic disturbance $w_{k}$. We assume
$w_{k}$ is an $n$-dimensional random vector defined in the probability space
$(\distspace, \sigma(\distspace),\mathbb{P}_{w})$; $\sigma( \distspace)$
denotes the minimal $\sigma$-algebra associated with the random vector $w_{k}$.
We assume the disturbance $w_{k}$ is absolutely continuous with a probability
density function $\psi_{w}$ and the disturbance process
${\{w_{k}\}}_{k=0}^{N-1}$ is an independent and identically distributed
(i.i.d.) random process. 
% The assumption that $w_{k}$ is absolute continuity 
% ensures the existence of a density function. 
% The given stochastic system is,
% \begin{equation}
%   x_{k+1} = f_{k}(x_{k}, u_{k}) + w_{k},\qquad w_{k}\sim \prob_{w}.
%   \label{eq:nonlin-stoch}
% \end{equation}

We will denote an admissible universally-measurable state-feedback law as $\mu:
\statespace \rightarrow \inputspace$  and the set of Markov control
policies
% \footnote{Note that we do not place a bar over $\pi$ to maintain
% consistency with previous works
% \cite{abate2008_reachability,summers2010_verification,summers2011stochastic}.}
$\pi=[\mu_0(\cdot), \mu_1(\cdot), \dots, \mu_{N-1}(\cdot)]$ as $ \mathcal{M}$.
Since no measurability restrictions were imposed on the admissible feedback laws
in Section \ref{sub:reachability-of-a-target-tube},
$\mathcal{M}\subset \mathcal{P}$. 
Given a Markov policy $\pi$ and initial state
$x_0\in \statespace$, the concatenated state vector $\bar{x}_{[1,N]} =
\transpose{[\transpose{x_{1}}, \dots, \transpose{x_{N}}]}$ for the system
\eqref{eq:nonlin} is a random vector defined in the probability space
$(\statespace^{N},\sigma(\statespace^{N}), \probxpi{0})$. The probability
measure $\probxpi{0}$ is induced from the sequence of random variables
$\bar{w}_{[0,N-1]}$ defined on the probability space $(\distspace^{N},
\sigma(\distspace^{N}), \probdist{0})$, as seen in~\cite[Sec.
2]{summers2010_verification}.
For $k\in \mathbb{N}_{[0,N-1]}$, we denote the probability space associated
with the random vector $\bar{x}_{[k+1,N]}$ as $(\statespace^{N-k},
\sigma(\statespace^{N-k}), \probxpi{k})$.

For stochastic reachability analysis, we are interested in the maximum
likelihood that the system \eqref{eq:nonlin} starting at an initial state
$x_0\in \statespace$ will stay within the target tube using a Markov policy.
The maximum likelihood and the optimal Markov policy can be determined as the
solution to the optimization problem, \cite[Sec. 4]{summers2010_verification}
\begin{equation} \label{eq:ra-value-fcn}
  \sup_{\pi\in \mathcal{M}} \mathbb{E}^{N,\pi}_{\bar{x}}\left[
    \prod_{i=0}^{N} \ONE_{\targetset_{i}}(x_{i}) \right].
\end{equation}
Let the optimal solution to problem \eqref{eq:ra-value-fcn} be
$\pi^\ast=[\mu_0^\ast(\cdot)\ \ldots\ \mu_{N-1}^\ast(\cdot)]$, the
\emph{maximal Markov policy in the terminal sense}~\cite[Def.
10]{summers2010_verification}. A dynamic programming approach was presented
in~\cite{summers2010_verification} to solve problem
\eqref{eq:ra-value-fcn}, along with sufficient conditions for the existence of
a maximal Markov policy. This approach computes value functions $V_{k}^{\ast}:
\statespace \rightarrow [0,1]$ for $k\in[0,N]$,
\begin{subequations}
    \begin{align}
        V_{N}^{\ast}(x) &= \ONE_{ \mathcal{T}_{N}}(x) \label{eq:valFunN} \\
        V_{k}^{\ast}(x)&= \ONE_{\targetset_{k}}(x)\int_{ \statespace} V_{k+1}^{\ast}(y)Q(dy|x, u) \nonumber\\
            &=\ONE_{ \targetset_{k}}(x)\probxpistar{k}\left(\bigcap_{t=k+1}^N \{x_{t} \in \targetset_{t}\} \Big\vert x_{k} = x \right)
            \label{eq:valFun}
    \end{align}
    \label{eq:dyn_recurs}
\end{subequations}
\kern-3pt where $Q(\cdot|x,u)$ is the one-step transition kernel
\cite{summers2010_verification}. For notational convenience we will often
simplify the $x_{k} = x$ condition in probability statements and simply write
$\probxpistar{k}\left(\cdot \vert x_{k} = x \right)$ as
$\probxpistar{k}\left(\cdot \vert x \right)$.

\begin{figure*}[t!]
    \makeatletter
    % \newcount=\eqfreeze
    % \eqfreeze=9
    \newcount\tmpcount
    \tmpcount=9
    \advance\tmpcount by -\c@equation
    \global\advance\c@equation by \tmpcount 
    \makeatother
    \begin{align}
        V_{k}^\ast(x)
        &= \probxpistar{k}\Bigg(\bigcap_{t=k+1}^N \{x_{t} \in \targetset_{t}\} \Big\vert x \Bigg)\ONE_{ \targetset_{k}}(x)  \nonumber \\
        &= \probxpistar{k}\Bigg(\bigcap_{t=k+1}^N \{x_{t} \in \targetset_{t}\} \Big\vert x, \bar{w}_{[k,N-1]} \in \boundeddist^{N-k} \Bigg)  \probdist{k}(\bar{w}_{[k,N-1]} \in \boundeddist^{N-k}) \nonumber \\
        &\qquad+\probxpistar{k}\Bigg(\bigcap_{t=k+1}^N \{x_{t} \in \targetset_{t}\} \Big\vert x, \bar{w}_{[k,N-1]} \in (\mathcal{W}^{N-k}\setminus\boundeddist^{N-k}) \Bigg)  \probdist{k}(\bar{w}_{[k,N-1]} \in
        (\mathcal{W}^{N-k}\setminus\boundeddist^{N-k}))\label{eq:total-prob-vf-expansion}\\
        &\geq  \probxpistar{k}\Bigg(\bigcap_{t=k+1}^N \{x_{t} \in \targetset_{t}\} \Big\vert x, \bar{w}_{[k,N-1]} \in \boundeddist^{N-k} \Bigg)  \probdist{k}(\bar{w}_{[k,N-1]}\in \boundeddist^{N-k}).\label{eq:lower-bound-vf} 
    \end{align}
    \hrulefill
    \makeatletter 
    \advance\tmpcount by 2
    \global\advance\c@equation by -\tmpcount 
    \makeatother
\end{figure*}

By definition, the optimal value function $V_{0}^{\ast}(x_0)$ provides the
maximum likelihood of ensuring that the system \eqref{eq:nonlin} stays
within the target tube $\targettube_{N}$ when initialized to the initial state
$x_0\in \statespace$. Note that the problem discussed
in~\cite{summers2010_verification} was specifically a reach-avoid problem, but
it may be easily extended to the more general case discussed here yielding the
recursion \eqref{eq:dyn_recurs}. 

For $\alpha\in[0,1]$ and $k\in
\mathbb{N}_{[0,N]}$, the \emph{stochastic $\alpha$-level reach set},
\begin{align}
    \stochlevelset{k} &= \Bigg\{ y\in\targetset_{k} : \nonumber \\
    &\hskip-1cm\sup_{\pi\in \mathcal{M}}\probxpi{k}\Bigg( \bigcap_{t=k+1}^N \{\bar{x}_{t} \in \targetset_{t}\} \; \Big\vert \; y \Bigg) \geq \alpha \Bigg\} \nonumber,\\
&= \left\{ y\in\targetset_{k} : V_{k}^{\ast}(x) \geq
\alpha \right\} \label{eq:sra-set} 
\end{align}
is the set of states $x$ that achieve the reachability of a target tube
objective with a minimum probability $\alpha$. This set is shown graphically
in Figure \ref{fig:stoch-reach-set}. For brevity we will often refer to the
stochastic $\alpha$-level reach set simply as the {\itshape stochastic reach
set}.

\subsection{Disturbance minimal and maximal reachability of a target tube}

Now, consider the nonlinear system, \eqref{eq:nonlin}, with an uncertain
disturbance $w_{k}$ drawn from a bounded disturbance set. Two types of
phenomena are of interest: the \emph{disturbance minimal reachability of a
target tube} and the \emph{disturbance maximal reachability of a target tube}.
In this subsection, we will treat the disturbance as a \emph{non-stochastic}
uncertainty. For these problems we also consider admissible state-feedback
policies, $\rho \in \mathcal{P}$.
% We will denote an admissible state-feedback law using $\nu:
% \statespace \rightarrow \inputspace$, the set of admissible state-feedback laws
% with $ \mathcal{F}$, and the set of admissible control policies $\rho =
% [\nu_0(\cdot), \nu_1(\cdot), \dots, \nu_{N-1}(\cdot)]$ using $ \mathcal{P} =
% \mathcal{F}^N$.

For the disturbance minimal reachability problem, 
% The problem of robust reachability of the target tube concerns with the system
% \eqref{eq:nonlin} with $w_{k}\in\boundeddist\subseteq \mathcal{W}$, \emph{i.e.},
% \begin{equation}
%   x_{k+1} = f_{k}(x_{k}, u_{k}) + w_{k},\qquad w_{k}\in \boundeddist.
%   \label{eq:nonlin}
% \end{equation}
we are interested in the existence of a feedback controller $\rho\in
\mathcal{P}$ for which the system \eqref{eq:nonlin} will remain in a target
tube, $\targettube_{N}$, despite the \emph{worst possible} choice of the
disturbance $w_{k} \in \boundeddist$. The set of states $x_{k}$ for $k\in
\mathbb{N}_{[0,N-1]}$ which may be driven by such a control policy is  the
\emph{$k$-time disturbance minimal reach set}
\begin{align}
    \reachunder{k} = \big\{ x_{k} \in \targetset_{k} : \; &\forall t \in \dtinterval{k}{N-1}, \exists \nu_{t} \in \mathcal{F}, \nonumber\\
        &\forall w_{t} \in \boundeddist, x_{t+1} \in \targetset_{t+1} \big\}
    \label{eq:dist-minimal-reach-set}
\end{align}
% The qualifier ``effective'' emphasizes that the sets
% $\reachunder{k}$ serve as the new effective target set
% at time $k$ since for any state in $\reachunder{k}$
% there exists a control such that the states $x_{t}$ will remain in the target
% tube despite worst-case disturbance effects for all $t \in \dtinterval{k}{N-1}$.
% as opposed to the original target set $ \targetset_{k}$. 
% The robust effective target sets have also been studied as the \emph{robust controllable set} in model predictive control literature~\cite[Defn 11.?]{BorelliTextbook}.
%For a state described by \eqref{eq:nonlin-dist}, we note that $x_{k} \in \targettube_{N}$ if $k \leq N$ and $x_{k} \in \targetset_{k}$. 

The problem of disturbance maximal reachability of the target tube concerns the
system \eqref{eq:nonlin} with $w_{k}\in\boundeddistover\subseteq \mathcal{W}$.
% \emph{i.e.},
% \begin{equation}
%   x_{k+1} = f_{k}(x_{k}, u_{k}) + w_{k},\qquad w_{k}\in \boundeddistover.
%   \label{eq:nonlin}
% \end{equation}
Here, we seek a feedback controller $\rho\in \mathcal{P}$ ensures that the
system stays within the target tube, under the \emph{best possible} choice for
the disturbance $w_{k}\in \mathcal{O}$. The set of states $x_{k}$ for $k\in
\mathbb{N}_{[0,N-1]}$ which may be driven by such a control policy is the
\emph{$k$-time disturbance maximal reach set}, 
\begin{align}
    \reachover{k} = \big\{ x_{k} \in \targetset_{k} : \; &\forall t \in \dtinterval{k}{N-1}, \exists \nu_{t} \in \mathcal{F}, \nonumber\\
        &\exists w_{t} \in \boundeddistover, x_{t+1} \in \targetset_{t+1} \big\}
    \label{eq:dist-maximal-reach-set}.
\end{align}
The disturbance minimal and maximal reach sets, \eqref{eq:dist-minimal-reach-set}
and \eqref{eq:dist-maximal-reach-set}, can be seen graphically in 
Figure \ref{fig:uncertain-reach-sets}.

For brevity, when the discrete-time index $k$ is apparent, we will refer to the
$k$-time disturbance minimal and maximal reach sets more succinctly as the
\emph{disturbance minimal} and \emph{disturbance maximal reach sets}. Note that the
\emph{disturbance minimal reach set} is equivalent to the well studied in reachability of
target tube problem~\cite{bertsekas1971minimax,bertsekasDP}, and is also known
as the \emph{robust controllable set} in model predictive control community
\cite{borelli_2017}. 
% In this paper, we add the qualifiers ``robust'' and
% ``augmented'' to distinguish the sets defined in
% \eqref{eq:dist-minimal-reach-set} and
% \eqref{eq:dist-maximal-reach-set}. 

% subsection Target tubes (end)

\subsection{Problem statements}

This paper utilizes disturbance minimal and maximal reach sets to efficiently
and scalably approximate the stochastic $\alpha$-level reach set. We achieve
this by addressing the following problems:
\begin{problem}%{1.a}
     Given a value $\alpha\in[0,1]$, characterize $\boundeddist,
     \boundeddistover \subseteq \distspace$ whose corresponding disturbance 
     minimal and maximal reach sets respectively under and overapproximate
     the stochastic effective $\alpha$-level set \eqref{eq:sra-set},
     \emph{i.e.}, find $\boundeddist, \boundeddistover \subseteq \distspace$ 
     such that $\reachunder{0} \subseteq
     \mathcal{L}_{0}(\targettube_{N},\alpha)\subseteq
     \reachover{0}$.
     % with the sets
     % corresponding to the systems \eqref{eq:nonlin},
     % \eqref{eq:nonlin}, and \eqref{eq:nonlin}
     % respectively.
     \label{problem:under-and-over-approx}
\end{problem}
\begin{customthm}{1.a}
    Characterize the sufficient conditions for which the optimal control policy
    corresponding to the disturbance minimal and maximal reach sets is a Markov 
    control policy for the system \eqref{eq:nonlin}.
    \label{problem:optimal-policy-is-markov}
\end{customthm}
Next, we discuss convex optimization-based techniques to construct $
\boundeddist$ and $ \boundeddistover$.
\begin{problem}
    Propose computationally efficient methods to compute $ \boundeddist,
    \boundeddistover\subseteq \mathcal{W}$ that satisfy
    Problem~\ref{problem:under-and-over-approx}.\label{prob:boundeddist}
\end{problem}
We will also discuss the min-max and min-min problems that may be used to
obtain these sets~\cite[Sec. 4.6.2]{bertsekasDP}.
\begin{problem}%{1.b}
    Construct Lagrangian-based recursions for the exact computation of the
    $k$-time disturbance minimal and disturbance maximal reach sets for the system \eqref{eq:nonlin}.
    \label{problem:lagrangian-recursions}
\end{problem}
\begin{problem}%{1.b}
    Improve the approximations obtained via Problem~\ref{problem:under-and-over-approx} by using
    multiple disturbance subsets.
    \label{problem:multiple-disturbance-sets}
\end{problem}

\section{Lagrangian approximations for the stochastic effective level sets}
\label{sec:stochastic-approx}

In this section we will details how disturbance minimal and maximal reach sets
are used to approximate stochastic reach sets. We will first
detail the conditions upon $\boundeddist$ which allow for the disturbance minimal
reach set to be a conservative underapproximation, and then will detail the
% \linebreak
% \vskip2in
% \noindent
conditions for $\boundeddistover$ that ensures that the disturbance maximal reach
set will provide an overapproximation. The theory in this section will
require the conditions stated in the following assumption.

\begin{assumption}
    The target tube $\targettube_{N}$ is closed, the input space $\inputspace$
    is compact, and $f$ is continuous in $\statespace \times
    \inputspace$.\label{assum:exist}
\end{assumption}
\subsection{Underapproximation of the stochastic $\alpha$-level reach set}

We will first address Problem~\ref{problem:optimal-policy-is-markov} via Theorem~\ref{thm:meas}.
% , whose
% proof is deferred to Section~\ref{sub:proof-meas}. The existence of a Markov
% policy permits the definition of a probability space associated with
% $\bar{x}_{[k,N-1]}$.

\begin{theorem}
    Under Assumption~\ref{assum:exist}, there exists an optimal Markov policy
    $\pi^\ast\in \mathcal{M}\subset \mathcal{P}$ associated with the set
    $\reachunder{0}$.
    \label{thm:meas}
\end{theorem}
The proof for Theorem \ref{thm:meas} is deferred to Section
\ref{sub:proof-meas}. The guarantee of the existence of an optimal Markov
policy for the robust effective target set allows for the definition of the
probability measure $\probxpi{k}(\cdot|y)$, which is essential for
demonstrating the conditions upon $\boundeddist$ that will make the disturbance
minimal reach set a conservative approximation of the stochastic reach
set, as will be shown in the following Proposition and Theorem.
\begin{proposition}\label{prop:Probab1}
    Under Assumption~\ref{assum:exist}, for every $k\in \dtinterval{0}{N-1}$,
    if $y \in \reachunder{k}$, then
    \begin{align}
        \probxpistar{k}\Bigg(\bigcap_{t=k+1}^N \{x_{t} \in \targetset_{t}\} \; \Big\vert \;  & y, \bar{w}_{[k,N-1]}\in \boundeddist^{N-k}\Bigg) = 1 \label{eq:prob-forall-equals-1}.
    \end{align}
    \label{prop:Prop1}
\end{proposition}
\begin{proof}
    Theorem \ref{thm:meas} ensures that the probability measure on the left-hand
    side of \eqref{eq:prob-forall-equals-1} exists. The equality is thus ensured
    by the definition of the disturbance minimal reach set, \eqref{eq:dist-minimal-reach-set}.
    % Follows from Theorem~\ref{thm:meas} and
    % \eqref{eq:dist-minimal-reach-set}. 
\end{proof}

\begin{theorem} \label{lem:bounded-conservative}
    Under Assumption~\ref{assum:exist}, for some $\alpha\in[0,1]$, $k \in
    \dtinterval{0}{N-1}$, and $\boundeddist \subseteq \mathcal{W}$ such that
    for all $t \in \dtinterval{k}{N-1}$, $\mathbb{P}_{w}(w_{t} \in \boundeddist)
    = {\alpha}^{\frac{1}{N-k}}$, then $\reachunder{k}
    \subseteq \mathcal{L}_{k}(\alpha)$.\label{thm:underapprox}
\end{theorem}
\begin{proof}
    We are interested in underapproximating $\stochlevelset{k}$ as defined
    in~\eqref{eq:sra-set}. If $x \in \reachunder{k}$, then Equation
    \ref{eq:total-prob-vf-expansion} follows from \eqref{eq:valFun} by Theorem \ref{thm:meas},
    the law of total probability, and the definition of the robust effective
    target set \eqref{eq:dist-minimal-reach-set}---which implies that
    $\ONE_{\targetset_{k}}(x) = 1$.
    Equation \eqref{eq:lower-bound-vf} follows from \eqref{eq:total-prob-vf-expansion} after
    ignoring the second term (which is non-negative). Simplifying
    \eqref{eq:lower-bound-vf} using Proposition~\ref{prop:Prop1} and the i.i.d.
    assumption of the disturbance process, we obtain
    \makeatletter \advance\c@equation by 2 \makeatother
    \begin{align}
        V_{k}^\ast(x) &\geq \probdist{k}(\bar{w}_{[k,N-1]} \in \boundeddist^{N-k}) \nonumber \\
        &=\left(\mathbb{P}_{w}(w_{t} \in \boundeddist)\right)^{N-k} =
        \alpha. \label{eq:approx}
    \end{align}
    Thus, if $x \in \reachunder{k}$ then $x \in \stochlevelset{k}$ by
    \eqref{eq:sra-set}, implying $\reachunder{k} \subseteq \stochlevelset{k}$.
\end{proof}

Theorem~\ref{lem:bounded-conservative} characterizes conditions upon
$\boundeddist$ such that $\reachunder{k}$ will be a conservative
underapproximation of $\stochlevelset{k}$. This will allow for fast
underapproximations of $\stochlevelset{k}$ to be determined through Lagrangian
computation of the robust effective target set. These methods will be further
detailed in Section
\ref{sec:lagrangian-methods-recursion}.

% subsection Underapproximation of the stochastic reachable $\alpha$-level set (end)
\begin{figure*}
    \makeatletter
    % \newcount=\eqfreeze
    % \eqfreeze=9
    \newcount\tmpcount
    \tmpcount=14
    \advance\tmpcount by -\c@equation
    \global\advance\c@equation by \tmpcount 
    \makeatother
    \begin{align}
        \probxpi{k}\left( \bigcup_{t=k+1}^N \{x_{t} \not\in \targetset_{t}\}\Big\vert y \right) &= \probxpi{k}\left( \bigcup_{t=k+1}^N \{x_{t} \not\in \targetset_{t}\}\Big\vert y, \bar{w}_{[k,N-1]} \in \boundeddistover^{N-k}\right) \probdist{k}( \bar{w}_{[k,N-1]}\in \boundeddistover^{N-k}) \nonumber \\
        &\hskip1cm+ \probxpi{k}\left( \bigcup_{t=k+1}^N \{x_{t} \not\in \targetset_{t}\}\Big\vert y, \bar{w}_{[k,N-1]} \not\in \boundeddistover^{N-k}\right) \probdist{k}(\bar{w}_{[k,N-1]} \not\in \boundeddistover^{N-k}) \nonumber \\ &\geq\probxpi{k}\left( \bigcup_{t=k+1}^N \{x_{t} \not\in \targetset_{t}\}\Big\vert y, \bar{w}_{[k,N-1]} \in \boundeddistover^{N-k} \right) \probdist{k}( \bar{w}_{[k,N-1]}\in \boundeddistover^{N-k}). \label{eq:reach-over-lower-bound}
    \end{align}
    \hrulefill
    \makeatletter 
    \advance\tmpcount by 1
    \global\advance\c@equation by -\tmpcount 
    \makeatother
    \hrulefill
\end{figure*}

\subsection{Overapproximation of the stochastic $\alpha$-level reach set}

We now move to the less common analysis of the overapproximation of the
stochastic $\alpha$-level reach set.

\begin{theorem}
    Under Assumption~\ref{assum:exist}, $\alpha\in[0,1]$, some  $k \in
    \dtinterval{0}{N}$, and $\boundeddistover \subseteq \mathcal{W}$ such that
    for all $t \in \dtinterval{0}{N-1}$, $\mathbb{P}_{w}(w_{t} \in
    \boundeddistover) = {(1-\alpha)}^{\frac{1}{N-k}}$, then 
    $\reachover{k} \subseteq
    \mathcal{L}_{k}(\alpha)$.\label{thm:overapprox}
\end{theorem}
\begin{proof}
    After characterizing $\statespace\setminus\stochlevelset{k}$ and
    $\statespace\setminus\reachover{k}$, we
    will show that, $
    y\in\statespace\setminus\reachover{k}
    \Rightarrow y \in \statespace\setminus\mathcal{L}_{k}(\alpha)$, for the
    given $ \boundeddistover$. This implies that
    $\statespace\setminus\reachover{k}
    \subseteq \statespace\setminus\mathcal{L}_{k}(\alpha) \Rightarrow
    \mathcal{L}_{k}(\alpha) \subseteq
    \reachover{k}$, as desired.  

    Using DeMorgan's law, we can write $\statespace \setminus
    \stochlevelset{k}$ as in Equation \ref{eq:sraSetV}. From the definition of
    \eqref{eq:dist-maximal-reach-set},
    \begin{align}
        \statespace\setminus\reachover{k} &= \big\{ x_{k} \in \statespace :
        \exists t \in \dtinterval{k}{N-1}, \forall \nu_{t}\in \mathcal{F}, \nonumber \\
        &\hskip2.1cm  \forall w_{t}\in \boundeddistover, x_{t+1} \not\in
        \targetset_{t+1} \big\}.
    \end{align}
    Hence, given $y\in \statespace\setminus\reachover{k}$, for any admissible control policies ($\rho\in \mathcal{P}$),
    and specifically any Markov control policies ($\pi\in \mathcal{M}$) since $
    \mathcal{M}\subset \mathcal{P}$, 
   \begin{align}
       \probxpi{k}\Bigg( \bigcup_{t=k+1}^N \{\bar{x}_{t} \not\in \targetset_{t}\} \; \Big\vert \; & y, \bar{w}_{[k,N-1]} \in \boundeddistover^{N-k}\Bigg)=1\label{eq:ars-set-probability-one}
   \end{align}
   By the law of total probability, we have \eqref{eq:reach-over-lower-bound}. Since
   $\mathbb{P}_{w}(w_{t} \in \boundeddistover) = (1 - \alpha)^{\frac{1}{N-k}},\
   \forall t\in \mathbb{N}_{[k,N-1]}$, and the disturbance $w_{t}$ is i.i.d.,
   \makeatletter \advance\c@equation by 1 \makeatother  
   \begin{align}
       \probdist{k}( \bar{w}_{[k,N-1]}\in \boundeddistover^{N-k})=1-\alpha\label{eq:probbarw_{o}ver}
   \end{align}
    From \eqref{eq:ars-set-probability-one}, \eqref{eq:reach-over-lower-bound}, and
    \eqref{eq:probbarw_{o}ver}, we conclude that for every $\pi\in \mathcal{M}$,
    $\probxpi{k}\left( \bigcup_{t=k+1}^N \{\bar{x}_{t} \not\in
    \targetset_{t}\}\Big\vert y\right)\geq 1 - \alpha$, implying $ y \in
    \statespace\setminus\mathcal{L}_{k}(\targettube_{N},\alpha)$ by \eqref{eq:sraSetV}.
\end{proof}
\begin{figure*}
    \begin{align}
    \statespace\setminus\stochlevelset{k}
    &=\left\{ y\in\targetset_{k} : \sup_{\pi\in \mathcal{M}} \probxpi{k}\left( \bigcap_{t=k+1}^N \{\bar{x}_{t} \in \targetset_{t}\}\Big\vert y\right) < \alpha \right\} \nonumber \\
    &=\left\{ y\in\targetset_{k} : \sup_{\pi\in \mathcal{M}}\left(1 -\probxpi{k}\left( \bigcup_{t=k+1}^N \{\bar{x}_{t} \in \targetset_{t}\}\Big\vert y \right)\right) < \alpha \right\} \nonumber \\
    %&=\left\{ y\in\targetset_{k} : 1 +\sup_{\pi\in \mathcal{M}}\left(-\probxpi{k}\left( \bigcup_{t=k+1}^N \{\bar{x}_{t} \in \targetset_{t}\}\Big\vert y \right)\right) < \alpha \right\} \nonumber \\
    &=\left\{y\in\targetset_{k} : 1 -\inf_{\pi\in \mathcal{M}}\probxpi{k}\left( \bigcup_{t=k+1}^N \{\bar{x}_{t} \in \targetset_{t}\}\Big\vert y \right) < \alpha \right\} \nonumber \\
                          &=\left\{ y\in\targetset_{k} : \inf_{\pi\in \mathcal{M}}\probxpi{k}\left( \bigcup_{t=k+1}^N \{\bar{x}_{t} \in \targetset_{t}\}\Big\vert y \right) > 1-\alpha \right\} \nonumber \\
                          &=\left\{ y\in\targetset_{k} : \forall \pi \in \mathcal{M}, \probxpi{k}\left( \bigcup_{t=k+1}^N \{\bar{x}_{t} \in \targetset_{t}\}\Big\vert y \right) > 1-\alpha \right\} \label{eq:sraSetV}.
                      %&=( \statespace\setminus \targetset_{k})\bigcup\Bigg\{ x_{k}\in\statespace : \forall \pi \in \mathcal{F}, \nonumber \\
                      %&\qquad\qquad\qquad\qquad\left.\prob^{\pi}_{\bar{x}_{k}}\left( \bigcup_{t=1}^N \{\bar{x}_{t} \in \targetset_{t}\}\right) > 1-\alpha \right\}  \label{eq:sraSetV}.
\end{align}
\hrulefill
\end{figure*}

Note that because of the $\forall \pi \in \mathcal{M}$ in \eqref{eq:sraSetV} we
can be assured that the measure $\probxpi{k}(\cdot|y)$ in
\eqref{eq:ars-set-probability-one} is well defined. This contrasts with Theorem
\ref{thm:underapprox} for which we needed to Theorem \ref{thm:meas} is
necessary to ensure that the probability measure in
\eqref{eq:prob-forall-equals-1}.

% For verification purposes, we are interested in approximating the stochastic
% reach set at $k=0$. 
We summarize the sufficient conditions put
forward in Theorems~\ref{thm:underapprox} and~\ref{thm:overapprox} in
Theorem~\ref{thm:subset-and-superset} to achieve the desired approximation,
which addresses Problem~\ref{problem:under-and-over-approx}, i.e. characterizes the conditions upon
$\boundeddist$ and $\boundeddistover$ that ensure that the robust and augmented
effective target sets will under and overapproximate $\stochlevelset{k}$,
respectively.
\begin{theorem}
    Under Assumption~\ref{assum:exist}, $\alpha\in[0,1]$, and
    $\boundeddist,\boundeddistover \subseteq \mathcal{W}$ such that for all $t
    \in \dtinterval{0}{N-1}$, $\mathbb{P}_{w}(w_{t} \in \boundeddist) =
    {\alpha}^{\frac{1}{N}}$ and $\mathbb{P}_{w}(w_{t} \in \boundeddistover) =
    {(1-\alpha)}^{\frac{1}{N}}$, then $\reachunder{k} \subseteq
    \mathcal{L}_{k}(\alpha)\subseteq\reachover{k}$ for $k \in \dtinterval{0}{N}$.
    \label{thm:subset-and-superset}
\end{theorem}
\begin{remark}
    The bounded disturbance sets, $\boundeddist$, $\boundeddistover$, which
    satisfy Theorem \ref{thm:subset-and-superset} are not unique.
    \label{remark:non-unique-bounded-sets}
\end{remark}
We will describe later, in Section
\ref{sec:computation_of_boundeddist_boundeddistover}, methodologies for
computing these bounded sets. As noted in Remark
\ref{remark:non-unique-bounded-sets}, these sets are not unique. However, the
algorithms that will be presented will provide efficient means for generating
$\boundeddist$ and $\boundeddistover$ for Gaussian and non-Gaussian
disturbances alike.

\section{Lagrangian methods for computation of disturbance minimal and maximal reach sets}
\label{sec:lagrangian-methods-recursion}

In this section we presume that we have bounded sets $\boundeddist$ and
$\boundeddistover$ which satisfy the conditions given in Theorem 
\ref{thm:subset-and-superset}. We now demonstrate convenient backward recursions
to compute \eqref{eq:dist-minimal-reach-set} and \eqref{eq:dist-maximal-reach-set}
using set operations.

For these recursions we will need to define, as in \cite{maidens_2013,GleasonCDC2017},
the unperturbed, one-step backward reach set from a set 
$\mathcal{S} \subseteq \statespace$ as $\onestepbackwardreach{k}(\mathcal{S})$.
Formally, for a nonlinear system \eqref{eq:nonlin}
\begin{align}
    % \onestepforwardreach{k}(x) &\triangleq \{ x^+\in \statespace : \exists u \in \inputspace,\ x^+=f_{k}(x,u)\}\label{eq:reachF}\\
    \onestepbackwardreach{k}(\mathcal{S})&\triangleq\left\{ x^- \in \statespace : \exists u\in
    \inputspace, \exists y\in \mathcal{S},\ y=f_{k}(x^-,u)\right\} \nonumber \\ %\label{eq:dr1-set}\\
    &=\left\{ x^- \in \statespace : \onestepforwardreach{k}(x^-)\cap \mathcal{S}\neq\emptyset
    \right\} \label{eq:dFcapS}
\end{align}
For an LTV system \eqref{eq:lin}, these can be written as 
\begin{align}
    % \onestepforwardreach{k}(x)&=A_{k}\{x\}\oplus B_{k} \inputspace,\label{eq:FOlin} \\
    \onestepbackwardreach{k}( \mathcal{S})&= A_{k}^{-1}(\mathcal{S}\oplus (-B_{k}
    \inputspace)).\label{eq:reachOlinCts}
\end{align}

\subsection{Recursion for disturbance minimal reach set $\reachunder{k}$}

% As in~\cite{maidens_2013,GleasonCDC2017}, we define the unperturbed, one-step
% forward reach set from a point $x\in \statespace$ for the system
% \eqref{eq:nonlin} as $\onestepforwardreach{k}(x)$, and the unperturbed, one-step backward reach
% set from a set $ \mathcal{S} \subseteq \statespace$ as $\onestepbackwardreach{k}(\mathcal{S})$.
% Formally, for the system \eqref{eq:nonlin},

% With these definitions we can establish a recursion for the exact computation
% of the robust effective target set.

The following theorem details the backward recursion for the disturbance minimal
reach set.

\begin{theorem}\label{thm:robust-effective-target-recursion}
  For the system given in (\ref{eq:nonlin}), the $k$-time disturbance minimal
  reach set $\reachunder{k}$ can be computed using the recursion  for
  $k\in \mathbb{N}, k < N$:
  \begin{align}
      \reachunder{N} &= \targetset_{N}
      \label{eq:dra_{r}ecurse0}\\
      \reachunder{k} &= \targetset_{k} \cap \onestepbackwardreach{k}(\reachunder{k+1} \minkdiff \boundeddist)\label{eq:min-reach-recursion}
  \end{align}
\end{theorem}
\begin{proof}
    We prove this by induction. Starting with the base case, $k=N-1$,
    \begin{align}
        \targetset_{N-1} \,\cap\, & \onestepbackwardreach{N-1}(\reachunder{N} \minkdiff \boundeddist) \nonumber \\
        & = \targetset_{N-1} \cap\, \onestepbackwardreach{N-1}(\targetset_{N} \minkdiff \boundeddist) \\
        & = \targetset_{N-1} \cap\, \big\{x^{-} \in \statespace : \exists \nu \in \mathcal{F}, \forall w_{N-1} \in \boundeddist, \nonumber \\
        & \hskip2cm f_{N-1}(x^{-}, \nu(x^{-})) + w_{N-1} \in \targetset_{N} \big\} \\
        & = \big\{ x_{N-1} \in \targetset_{N-1}: \exists \nu_{N-1} \in \mathcal{F}, \forall w_{N-1} \in \boundeddist, \nonumber \\
        & \hskip2cm f_{N-1}(x_{N-1}, \nu(x_{N-1})) + w_{N-1} \in \targetset_{N} \big\} \\
        & = \big\{ x_{N-1} \in \targetset_{N-1} : \exists \nu_{N-1} \in \mathcal{F}, \forall w_{N-1} \in \boundeddist, \nonumber \\
        & \hskip3.2cm x_{N} \in \targetset_{N} \big\} \\
        & = \reachunder{N-1}.
    \end{align}
    For any $k \in \naturalnums$, $k < N-1$
    \begin{align}
        \targetset_{k} \,\cap\, &\onestepbackwardreach{k}(\reachunder{k+1} \minkdiff \boundeddist) \nonumber \\
        & = \big\{ x_{k} \in \targetset_{k}: \exists \nu_{k} \in \mathcal{F}, \forall w_{k} \in \boundeddist, \nonumber \\
        & \hskip1.5cm f_{k}(x_{k}, \nu(x_{k})) + w_{k} \in \reachunder{k+1} \big\} \label{eq:recursion-k-start}\\
        & = \big\{ x_{k} \in \targetset_{k} : \exists \nu_{k} \in \mathcal{F}, \forall w_{k} \in \boundeddist,  \nonumber \\
        & \hskip1.5cm x_{k+1} \in \reachunder{k+1} \big\}.
    \end{align}
    By expanding $\reachunder{k+1}$ with its definition
    \eqref{eq:dist-minimal-reach-set},
    \begin{align}
         &= \big\{  x_{k} \in \targetset_{k} : \forall t \in \dtinterval{k}{N-1}, \exists \nu_{t} \in \mathcal{F}, \hskip0.6cm \nonumber \\
        & \hskip2.04cm \forall w_{t} \in \boundeddist, x_{t+1} \in \targetset_{t+1} \big\} \nonumber \\
        & = \reachunder{k} \label{eq:recursion-k-end}
    \end{align}
    which completes the proof.
\end{proof}

In systems for which the disturbance is not affine, i.e.
\begin{equation*}
    x_{k+1} = f_{k}(x_{k}, u_{k}, w_{k})
\end{equation*}
the recursion in Theorem \ref{thm:robust-effective-target-recursion} can be
altered by replacing $\reachunder{k+1} \minkdiff \boundeddist$ with
$\mathrm{Pre}(\reachunder{k+1})$, where $\mathrm{Pre}(\cdot)$ is the predecessor
set \cite{rakovic2006_reach}. Since many systems can be written in the form
\eqref{eq:nonlin}, i.e. with affine disturbances, we do not formally prove
the recursion using predecessor sets. 

\subsection{Recursion for disturbance maximal reach set $\reachover{k}$}

We follow a similar methodology as in the previous section to establish a
recursion for computing the disturbance maximal reach set.

\begin{figure*}
    \begin{center}
        \includegraphics{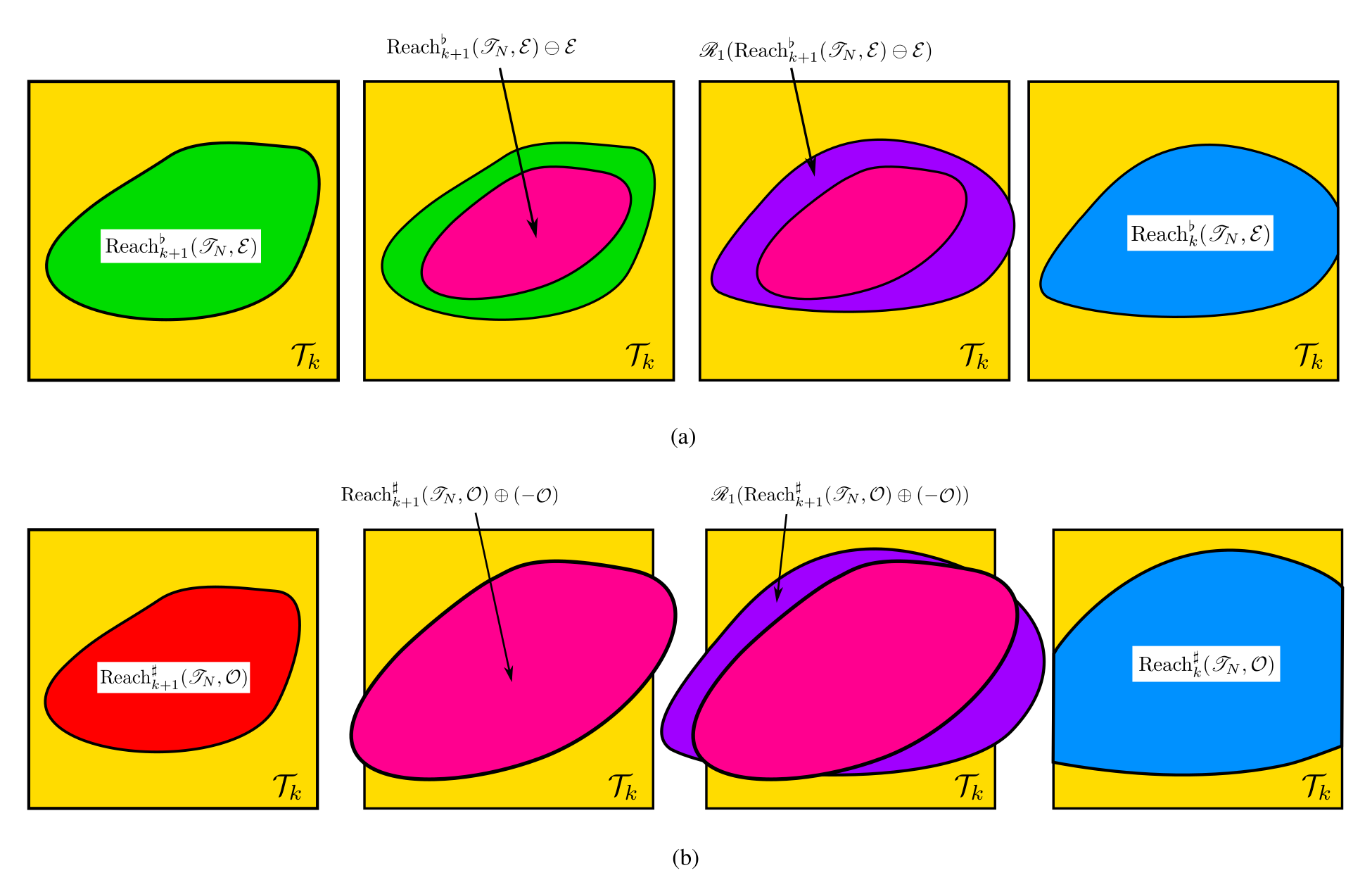}
    \end{center}
    \caption{Graphical depiction of Lagrangian recursion methods for computing (a) the
    $k$-time disturbance minimal reach set, $\reachunder{k}$, from
    $\reachunder{k+1}$ via \eqref{eq:min-reach-recursion} and (b) the $k$-time
    disturbance maximal reach set, $\reachover{k}$ from $\reachover{k+1}$ via
    \eqref{eq:max-reach-recursion}.}
    \label{fig:reach-sets-recursion}

    \makeatletter
    \@bsphack 
    \protected@write \@auxout {}{\string \newlabel {fig:min-reach-set-recursion}{{\@currentlabel (a)}{\thepage }}}
    \protected@write \@auxout {}{\string \newlabel {fig:max-reach-set-recursion}{{\@currentlabel (b)}{\thepage }}}
    \@esphack
\end{figure*}

\begin{theorem}\label{thm:augmented-effective-target-recursion}
    For the system given in (\ref{eq:nonlin}), the $k$-time
    disturbance maximal reach set $\reachover{k}$ can be computed using the
    recursion for $k \in \mathbb{N}, k < N$:
    \begin{align}
        \reachover{N} &= \targetset_{N} 
        \label{eq:raover_{r}ecurse0} \\
        \reachover{k} &= \targetset_{k} \cap
        \onestepbackwardreach{k}(\reachover{k+1} \minkadd (-\boundeddistover))\label{eq:max-reach-recursion}
    \end{align}
\end{theorem}
\begin{proof}
    Again, we prove this by induction. Starting with the base case, $k=N-1$,
    \begin{align}
        \targetset_{N-1} \,\cap\, & \onestepbackwardreach{N-1}(\reachover{N} \minkadd (- \boundeddistover)) \nonumber \\
        & = \targetset_{N-1} \cap\, \onestepbackwardreach{N-1}(\targetset_{N} \minkadd (- \boundeddistover)) \\
        & = \targetset_{N-1} \cap\, \big\{x^{-} \in \statespace : \exists \nu \in \mathcal{F}, \exists w_{N-1} \in \boundeddistover \nonumber \\
        & \hskip2cm f_{N-1}(x^{-}, \nu(x^{-})) + w_{N-1} \in \targetset_{N} \big\} \\
        & = \big\{ x_{N-1} \in \targetset_{N-1}: \exists \nu_{N-1} \in \mathcal{F}, \exists w_{N-1} \in \boundeddistover, \nonumber \\
        & \hskip2cm f_{N-1}(x_{N-1}, \nu(x_{N-1})) + w_{N-1} \in \targetset_{N} \big\} \\
        & = \big\{ x_{N-1} \in \targetset_{N-1} : \exists \nu_{N-1} \in \mathcal{F}, \exists w_{N-1} \in \boundeddistover, \nonumber \\
        & \hskip3.2cm x_{N} \in \targetset_{N} \big\} \\
        & = \reachover{N-1}.
    \end{align}
    For any $k \in \naturalnums$, $k < N-1$
    \begin{align}
        \targetset_{k} \,\cap\, &\onestepbackwardreach{k}(\reachover{k+1} \minkadd (- \boundeddistover)) \nonumber \\
        & = \big\{ x_{k} \in \targetset_{k}: \exists \nu_{k} \in \mathcal{F}, \exists w_{k} \in \boundeddistover, \nonumber \\
        & \hskip1.5cm f_{k}(x_{k}, \nu(x_{k})) + w_{k} \in \reachover{k+1} \big\} \label{eq:recursion-k-start}\\
        & = \big\{ x_{k} \in \targetset_{k} : \exists \nu_{k} \in \mathcal{F}, \exists w_{k} \in \boundeddistover,  \nonumber \\
        & \hskip1.5cm x_{k+1} \in \reachover{k+1} \big\}.
    \end{align}
    By expanding $\reachover{k+1}$ with its definition \eqref{eq:dist-minimal-reach-set},
    \begin{align}
         &= \big\{  x_{k} \in \targetset_{k} : \forall t \in \dtinterval{k}{N-1}, \exists \nu_{t} \in \mathcal{F}, \hskip0.6cm \nonumber \\
        & \hskip2.04cm \exists w_{t} \in \boundeddistover, x_{t+1} \in \targetset_{t+1} \big\} \nonumber \\
        & = \reachover{k} \label{eq:recursion-k-end}
    \end{align}
    which completes the proof.
\end{proof}

Figure \ref{fig:reach-sets-recursion} shows, graphically, the recursion process 
for the disturbance minimal and maximal reach sets.

We synthesize the recursions shown in Theorems \ref{thm:underapprox} and
\ref{thm:overapprox} into algorithmic forms, see Algorithms \ref{alg:rets} and
\ref{alg:aets}. These algorithms compute the robust and augmented effective
target sets which have been shown to approximate the stochastic effective
$\alpha$-level set. 

\begin{algorithm}[]
    \SetKwInOut{Input}{Input}\SetKwInOut{Output}{Output}

    \Input{Target tube $\targettube_{N}$; system dynamics; desired probability level $\alpha \in [0,1]$; a bounded
    disturbance sets, $\boundeddist$; horizon length, $N$}
    \Output{Augmented effective target set, $\reachunder{0}$}
    \vskip6pt
    $\reachunder{N} \leftarrow \mathcal{T}_{N}$, \\
    \For{$i = N-1, N-2, \dots, 0$}{
        // from \eqref{eq:min-reach-recursion} \\
        $S \leftarrow \reachunder{i+1} \minkdiff \boundeddist$ \\
        $E \leftarrow \onestepbackwardreach{k}(S)$ \\
        $\reachunder{i} \leftarrow \mathcal{T}_{i} \cap E$ \hfill \\
    }
    \caption{Recursion for the disturbance minimal reach set.}
    \label{alg:rets}
\end{algorithm}
\begin{algorithm}[]
    \SetKwInOut{Input}{Input}\SetKwInOut{Output}{Output}

    \Input{Target tube $\targettube_{N}$; system dynamics; desired probability level $\alpha \in [0,1]$; a bounded
    disturbance sets, $\boundeddistover$; horizon length, $N$}
    \Output{Robust effective target set, $\reachover{0}$}
    \vskip6pt
    $\reachover{N} \leftarrow \mathcal{T}_{N}$, \\
    \For{$i = N-1, N-2, \dots, 0$}{
        // from \eqref{eq:max-reach-recursion} \\
        $S \leftarrow \reachover{i+1} \minkadd (- \boundeddistover)$ \\
        $E \leftarrow \onestepbackwardreach{k}(S)$ \\
        $\reachover{i} \leftarrow \mathcal{T}_{i} \cap E$ \hfill \\
    }
    \caption{Recursion for the disturbance maximal reach set.}
    \label{alg:aets}
\end{algorithm}

\subsection{Min-max formulation for $\reachunder{k}$}
\label{sub:min-max}

A min-max optimal control problem was presented in~\cite[Sec.
1]{bertsekas1971minimax},~\cite[Sec. 4.6.2]{bertsekasDP} to compute the disturbance
minimal reach set for the system \eqref{eq:nonlin}. The optimization
problem is:\\
\begin{subequations}
    \emph{(Min-max problem for robust effective target sets)}\vspace*{-0.5em}
\begin{align}
    &\hspace{-1em}\begin{array}{rl}
      \underset{\rho\in \mathcal{P}}{\mbox{minimize}}&\hskip-0.2cm\underset{\bar{w}_{[0,N-1]} \in \boundeddist^N }{\mbox{
      maximize}}\  J(\rho,\bar{w}_{[0,N-1]};x_0)\\
      \mbox{subject to}& \hskip-0.3cm\left\{\arraycolsep=2pt \begin{array}{rll}
              x_{k+1} &= f(x_{k},\nu_{k}(x_{k}))  +w_{k},&\quad k \in \mathbb{N}_{[0,N-1]}\\
      \nu_{k}(\cdot) & \in \mathcal{F} &\quad k \in \mathbb{N}_{[0,N-1]}\\
      \rho&=[\nu_0(\cdot)\ \ldots\ \nu_{N-1}(\cdot)]\\
      \bar{w}_{[0,N-1]} &=[w_0\ \ldots\ w_{N-1}(\cdot)]\\
    \end{array}\right.
  \end{array}
  \label{eq:prob-minmax-opt}\\
  &\hspace{7.5em}g_{k}(x_{k})\triangleq 1-\ONE_{\mathcal{T}_{k}}(x_{k}),\ k\in \dtinterval{0}{N} \label{eq:gk_{d}efn}\\
  &J(\rho,\bar{w}_{[0,N-1]};x_0)=\sum\limits_{k=0}^{N} g_{k}(x_{k}) \label{eq:J_{d}efn}
\end{align}\label{eq:prob-minmax}%
\end{subequations}%
where the decision variables are $\rho$ and $\bar{w}_{[0,N-1]}$. The min-max
optimal control problem \eqref{eq:prob-minmax} can be solved using dynamic
programming~\cite[Sec 1.6]{bertsekasDP}. We generate the cost-to-go/value
functions   $J^\ast_{k}(x): \mathcal{X} \rightarrow \dtinterval{0}{N-k+1}$ for
$k\in \dtinterval{0}{N}$ with the optimal value of problem
\eqref{eq:prob-minmax} when starting at $x_0$ as $J^\ast_0(x_0)$. We define
$J^\ast_{N}(x)=g_{N}(x)$, and obtain the remaining value functions via a
dynamic programming recursion,
\begin{subequations}
    \begin{align}
        H^\ast_{k}(u,x)&=\sup_{w\in \boundeddist}
    \left[J^\ast_{k+1}(f(x,u)+w)+g_{k}(x)\right]\label{eq:minmaxVH}\\
    J^\ast_{k}(x)&=\inf_{u\in \inputspace}H^\ast_{k}(u,x). \label{eq:minmaxV}
    \end{align}\label{eq:prob-minmax-recurs}%
\end{subequations}%
Note that here $J^\ast_{k}(x)$ records the minimum count of violations of the
target tube constraint by the system \eqref{eq:nonlin} when starting at $x$
at time $k$ under the optimal choice of the inputs $\nu_{k}(\cdot)$ and
adversarial choice of the disturbances $w_{k}$~\cite[Sec.  4.6.2]{bertsekasDP}.
This follows from \eqref{eq:gk_{d}efn}, where $x_{k}\not\in \targetset_{k}$ if and
only if $ g_{k}(x_{k})=1$ and $x_{k}\in \targetset_{k}$ if and only if $g_{k}(x_{k})=0$. By
construction,
\begin{align}
    \reachunder{0} &= \{x\in \mathcal{X}: J^\ast_0(x_0)=0\}.
    \label{eq:link-rob}
\end{align}

\begin{remark}
    A min-min problem can be similarly constructed for the disturbance maximal
    reach set for system \eqref{eq:nonlin}. Here, the system is driven to
    stay within the target tube with the input's optimal efforts augmented by
    the disturbance.
\end{remark}

\subsection{Proof of Theorem~\ref{thm:meas}}
\label{sub:proof-meas}

Theorem~\ref{thm:meas} states for systems \eqref{eq:nonlin} that satisfies
Assumption~\ref{assum:exist}, there exists an optimal Markov policy $\pi^\ast$
associated with the set $\reachunder{0}$, \emph{i.e.}.
we want to show that there is an optimal policy for \eqref{eq:prob-minmax}
$\pi^\ast\in \mathcal{M}\subset \mathcal{P}$. We will prove this using the
equivalent min-max problem formulated in Section~\ref{sub:min-max}. 

The organization of the proof is as follows: first, we will show that $g_{k}$ are
\lsc{} over $\statespace$. Then, we will show that, for every $k\in
\mathbb{Z}_{[0,N-1]}$, the functions $J_{k}^\ast$ and $H_{k}^\ast$ of
\eqref{eq:prob-minmax} are \lsc{} over $\statespace$ and
$\statespace\times\inputspace$ respectively, and that there exists a
Borel-measurable (and therefore universally measurable~\cite[Definition
7.20]{BertsekasSOC1978}) state-feedback control law $\mu_{k}^\ast(\cdot)$. We
thus construct an optimal Markov policy $\pi^\ast\triangleq[\mu_0^\ast(\cdot)\
\ldots\ \mu_{N-1}^\ast(\cdot)]$ associated with
$\reachunder{0}$, completing the proof. 

Since the target tube $\targettube_{N}$ is closed, $-\ONE_{\targetset_{k}}$ is
\lsc{} over $\statespace$ by (P1). We conclude that $g_{k}$ is \lsc{} over
$\statespace$ for $k\in \mathbb{Z}_{[0,t]}$ by the fact that constant functions
are \lsc{}\footnote{The set $\{x\in \mathcal{X}: 1\leq \lambda\}$ is either
empty ($\lambda<1$) or the entire $ \mathcal{X}$  ($\lambda\geq 1$), both of
which are closed.}, (P2), and \eqref{eq:gk_{d}efn}. 

Next, we prove the \lsc{} property of $J_{k}^\ast$ and $H_{k}^\ast$ and the
existence of a Borel-measurable $\nu_{k}^\ast(\cdot)$ by induction. Consider the
base case $k=N-1$. From \eqref{eq:minmaxVH}, 
\begin{align}
    H_{N-1}^\ast(u,x)&=\sup_{w_{N-1}\in \boundeddist} [J_{N}(f(x,u)+w_{N-1})+g_{N-1}(x)] \nonumber \\
                     &=\sup_{w_{N-1}\in \boundeddist} [g_{N}(f(x,u)+w_{N-1})+g_{N-1}(x)]\label{eq:minmaxVH_{N}minus1}
\end{align}
Since $f$ is continuous over $ \statespace\times \inputspace$,
$g_{N}(f(x,u)+w_{N-1})$ is \lsc{} over
$\statespace\times\inputspace\times\mathcal{W}$ by the fact that $g_{N}$ is
\lsc{} over $ \mathcal{X}$ and (P3). This implies that the objective in
\eqref{eq:minmaxVH_{N}minus1} is \lsc{} by the fact that $g_{N-1}$ is \lsc{} over
$\statespace$ and (P2). Hence, $H_{N-1}^\ast(u,x)$ is \lsc{} over
$\statespace\times\inputspace$ by (P4). Additionally, since $\inputspace$ is
compact and $H_{N-1}^\ast(u,x)$ is \lsc{}, $J_{N-1}^\ast(x)$ is \lsc{} over $
\mathcal{X}$ and there exists a Borel-measurable state-feedback law
$\mu_{N-1}(x)$ that optimizes \eqref{eq:minmaxV} by (P5). This completes the
proof of the base case.

Assume, for induction, the case $k=t,\ t\in \mathbb{N}_{[0,N-2]}$ is true, i.e,
$J_{t}^\ast$ is lower semicontinuous. Then, by the same arguments as above, we
conclude that $J_{t-1}^\ast$ and $H_{t-1}^\ast$ is \lsc{} over $\statespace$
and $\statespace\times \inputspace$, and a Borel-measurable state-feedback law
$\mu_{t-1}^\ast(x)$ exists via (P2)--(P5), and \eqref{eq:prob-minmax-recurs}.
This completes the induction, and demonstrates the existence of $\pi^\ast$
associated with $\reachunder{0}$.

% subsection Proof of Theorem~\ref{thm:meas} (end)

% section lagrangian_{m}ethods_{f}or_{r}obust_{a}nd_{a}ugmented_{e}ffective_{t}arget_{s}et_{c}omputation (end)

\section{Improving Approximations With Multiple Bounded Disturbance Sets}
\label{sec:algorithm_and_challenges}

Theorem \ref{thm:subset-and-superset} demonstrates that the disturbance minimal
and maximal reach sets are a under and overapproximation, respectively, 
of the stochastic reach set. These sets are computed using bounded disturbance
sets $\boundeddist$, and $\boundeddistover$. These disturbance sets, however,
are not unique (Remark 1) and many different disturbance sets can satisfy the 
sufficient conditions established in Section \ref{sec:stochastic-approx}.

In this section we will demonstrate how we can use many different bounded
disturbance sets $\boundeddist_{i}$, $\boundeddistover_{i}$, $i \in \{1, \dots,
M\}$, we can combine the these reach sets to help improve the approximation.
For this we will assume that $\boundeddist_{i}$, $\boundeddistover_{i}$ satisfy
the conditions established in Theorem \ref{thm:subset-and-superset} for each
$i \in \{1, \dots, M\}$. We will denote the disturbance minimal and maximal reach
sets using $\boundeddist_{i}$, $\boundeddistover_{i}$ as $\multreachunder{k}{i}$
and $\multreachover{k}{i}$, respectively.

% In the subsequent, we will detail a method to improve these
% approximations, i.e. reduce the conservativeness of the underapproximation and
% decrease the size of the overapproximation. 
% We improve the approximation by computing robust and effective target sets for
% many different bounded disturbance sets and combining the results in such a way
% as to maintain a consistent under or overapproximation. These improvements come
% at the increase in computational cost since we must compute multiple robust and
% augmented effective target sets, however, because these sets are determined via
% Lagrangian recursions the computation time is still small compared to other
% methods, such as dynamic programming.

% subsection Overall algorithms (end)

\subsection{Tighter approximations via multiple disturbance sets}

We first examine how to improve the underapproximation using multiple disturbance
minimal reach sets computed using $\boundeddist_{i}$. Algorithm \ref{alg:union-algorithm}, Theorems
\ref{thm:approx-unions} and \ref{thm:convex-hull}, and Lemma
\ref{lem:multiple-rets-superset} demonstrate these techniques.

\begin{theorem}
    \label{thm:approx-unions}
    Let $\boundeddist_{i}$, be a bounded set which 
    satisfies the condition
    \begin{equation}
        \prob\left(\bar{w}_{[k,N-1]} \in \boundeddist_{i}^{N-k} \right) = \alpha
    \end{equation}
    for all $i \in \naturalnums_{[1,M]}$, $\bar{w}_{[k,N-1]} = [w_{k}, w_{k+1},
    \dots w_{N-1}]$. For a nonlinear system \eqref{eq:nonlin}, the union
    of each disturbance minimal reach set is a subset of the true stochastic reach
    set, i.e.
    \begin{equation}
        \radnl{k}{[1,M]} = \bigcup_{m=1}^{M} \multreachunder{k}{m} \subseteq \levelset_{k}.
        \label{eq:union-subset}
    \end{equation}
\end{theorem}
\begin{proof}
    From \ref{thm:subset-and-superset} $\rad_{k}(\targettube_{N},
    \boundeddist_{i}) \subseteq \levelset_{k}$ for each $i \in
    \integers_{[1,N]}$. Thus the union of these sets remains a subset of
    $\levelset_{k}$.
\end{proof}

\begin{theorem}
    \label{thm:convex-hull}
    Let $\boundeddist_{i}$, be a bounded set which 
    satisfies the condition
    \begin{equation}
        \prob\left(\bar{w}_{[k,N-1]} \in \boundeddist_{i}^{N-k} \right) = \alpha
    \end{equation}
    for all $i \in \naturalnums_{[1,M]}$, $\bar{w}_{[k,N-1]} = [w_{k}, w_{k+1},
    \dots w_{N-1}]$. For a linear system \eqref{eq:lin}, if $\mathcal{U}$ is
    convex and compact, $\targettube_{N}$ is closed and convex, and $Q(\cdot|x,
    u)$ is continuous and log-concave, then the convex hull of the robust reach
    avoid set for each bounded disturbance $\rad_{k}(\targettube_{N},
    \boundeddist_{i})$, is a subset of the true reach-avoid level set
    $\levelset_{k}$, i.e. 
    \begin{equation}
        \radlin{k}{[1,M]} = \mathrm{Co}\big(\bigcup_{m=1}^{M} \multreachunder{k}{m} \big) \subseteq \levelset_{k}.
        \label{eq:convex-hull-subset}
    \end{equation}
\end{theorem}
\begin{proof}
    From \ref{thm:subset-and-superset} $\multreachunder{k}{i} \subseteq \levelset_{k}$ for each $i \in
    \naturalnums_{[1,M]}$, and from \cite[Theorem 4]{VinodHSCC2018},
    $\levelset_{k}$ is convex. Thus, the convex hull of disturbance minimal reach sets
    is a subset of $\levelset_{k}$ \cite[Section 2.3.4]{BoydConvex2004}.
\end{proof}

\begin{lemma}
    Let $\boundeddist_{j}$ be a bounded set in the collection $\boundeddist_{1}, \dots,
    \boundeddist_{M}$, $M,j \in \naturalnums$, $j < M$. Then
    \begin{equation}
        \multreachunder{k}{j} \subseteq \radlin{k}{[1,M]}
    \end{equation}
    \label{lem:multiple-rets-superset}
\end{lemma}
\begin{proof}
    For a nonlinear system \eqref{eq:nonlin}, for any $x \in
    \multreachunder{k}{j}$, $x \in \radnl{k}{[1,M]}$ by
    \eqref{eq:union-subset}, and for a linear system \eqref{eq:lin} any $y \in
    \multreachunder{k}{j}$, $y \in \radlin{k}{[1,M]}$ by
    \eqref{eq:convex-hull-subset}.
\end{proof}
\begin{algorithm}[]
    \SetKwInOut{Input}{Input}\SetKwInOut{Output}{Output}

    \Input{Target tube $\targettube$; system dynamics; desired probability
           level $\alpha \in [0,1]$; $M$ bounded disturbance sets,
           $\boundeddist_{1},\dots,\boundeddist_{M}$; horizon length, $N$}
    \Output{$N$-time stochastic reach-avoid $\alpha$-level set
            underapproximation, $\radnl{0}{[1,M]}$}
    \vskip6pt

    // Initialization \\
    \For{$m = 1, 2, \dots, M$}{
        $\multreachunder{N}{m} \leftarrow \mathcal{T}_{N}$
    }

    \vskip6pt
    // Recursion \\
    % $\radimproved{N}{[1,M]} \leftarrow \mathcal{T}$ \\
    \For{$i = N-1, N-2, \dots, 0$}{
        \For{$m = 1, 2, \dots, M$}{
            // from Algorithm \ref{alg:rets} \\
            $S \leftarrow \multreachunder{i+1}{m} \minkdiff \boundeddist_{m}$ \\
            $E \leftarrow \onestepbackwardreach{k}(S)$ \\
            $\multreachunder{i}{m} \gets \mathcal{T}_{i} \cap E$ \hfill \\
        }
        
        \If{nonlinear system dynamics \eqref{eq:nonlin}}{
            $\radnl{i}{[1,M]} \leftarrow 
        \bigcup_{m=1}^{M}\multreachunder{i}{m})$
        }

        \If{linear system dynamics \eqref{eq:lin}}{
            $\radlin{i}{[1,M]} \leftarrow 
        \mathrm{Co}\big(\bigcup_{m=1}^{M}\multreachunder{i}{m}\big)$
        }
    }
    \caption{Underapproximation of the stochastic effective $\alpha$-level set
             using multiple bounded disturbances and unions or convex hulls.}
    \label{alg:union-algorithm}
\end{algorithm}

Now we apply similar methodology to demonstrate how to improve the
overapproximation using multiple augmented effective target sets. 
\begin{theorem}
    \label{thm:approx-intersections}
    Let $\boundeddistover_{i}$, be a bounded set which 
    satisfies the condition
    \begin{equation}
        \prob\left(\bar{w}_{[k,N-1]} \in \boundeddistover_{i}^{N-k} \right) = \alpha
    \end{equation}
    for all $i \in \naturalnums_{[1,M]}$, $\bar{w}_{[k,N-1]} = [w_{k}, w_{k+1},
    \dots w_{N-1}]$. For a nonlinear system \eqref{eq:nonlin}, the union
    of each robust reach avoid set is a subset of the true reach-avoid level
    set, i.e.
    \begin{equation}
        \radovernl{k}{[1,M]} = \bigcap_{m=1}^{M} \multreachover{k}{m} \supseteq \levelset_{k}.
        \label{eq:intersection-subset}
    \end{equation}
\end{theorem}
\begin{proof}
    From \ref{thm:subset-and-superset} $\radover_{k}(\targettube_{N},
    \boundeddistover_{i}) \supseteq \levelset_{k}$ for each $i \in
    \integers_{[1,N]}$. Thus the intersection of these sets remains a subset
    of $\levelset_{k}$.
\end{proof}

\begin{lemma}
    \label{lem:intersections-supset}
    Let $\boundeddistover_{j}$ be a bounded set in the collection
    $\boundeddistover_{1}, \dots, \boundeddistover_{M}$, $M,j \in
    \naturalnums$, $j < M$. Then
    \begin{equation}
         \radovernl{k}{[1,M]} \subseteq \radover_{k}(\targettube_{N}, \boundeddistover_{j})
    \end{equation}
\end{lemma}
\begin{proof}
    From the intersection operation in \eqref{eq:intersection-subset}, clearly
    $\radovernl{k}{[1,M]} \subseteq \radover_{k}(\targettube_{N},
    \boundeddistover_{j})$ for any $j \in \{1, \dots, M\}$.
\end{proof}

The methods from Theorem \ref{thm:approx-intersections} and Lemma
\ref{lem:intersections-supset} are succinctly combined in Algorithm
\ref{alg:interect-over-approx}.

\begin{algorithm}[]
    \SetKwInOut{Input}{Input}\SetKwInOut{Output}{Output}

    \Input{Target tube $\targettube$; system dynamics; desired probability
           level $\alpha \in [0,1]$; $M$ bounded disturbance sets,
           $\boundeddist_{1},\dots,\boundeddist_{M}$; horizon length, $N$}
    \Output{$N$-time stochastic reach-avoid $\alpha$-level set
            underapproximation, $\radovernl{0}{[1,M]}$}
    \vskip6pt
    // Initialization \\
    \For{$m = 1, 2, \dots, M$}{
        $\multreachover{N}{m} \leftarrow \mathcal{T}_{N}$
    }

    \vskip6pt
    // Recursion \\
    % $\radover_{N}(\targettube_{N}, \boundeddist_{m}) \leftarrow \mathcal{T}_{N}$, $\forall m \in \mathbb{N}_{[1,M]}$ \\
    % $\radoverimproved{N}{[1,M]} \leftarrow \mathcal{T}$ \\
    \For{$i = N-1, N-2, \dots, 0$}{
        \For{$m = 1, 2, \dots, M$}{
            // from Algorithm \ref{alg:aets} \\
            $S \leftarrow \multreachover{i+1}{m} \minkadd (- \boundeddistover_{m})$ \\
            $E \leftarrow \onestepbackwardreach{k}(S)$ \\
            $\multreachover{i}{m} \leftarrow \mathcal{T}_{k} \cap E$ \hfill \\
        }
        
        \If{nonlinear system dynamics \eqref{eq:nonlin}}{
            $\radovernl{i}{[1,M]} \leftarrow 
        \bigcap_{m=1}^{M}\multreachover{i}{m}$
        }
    }
    \caption{Overapproximation of the stochastic effective $\alpha$-level set
             using multiple bounded disturbances and intersections.}
    \label{alg:interect-over-approx}
\end{algorithm}

% subsection improving_{t}he_{a}ugmented_{e}ffective_{t}arget_{s}et (end)

% subsection improving_{a}pproximations_{w}ith_{m}ultiple_{b}ounded_{s}ets (end)

% subsection Tightening the approximation using multiple disturbance sets (end)

\section{Computation of disturbance subsets: $\boundeddist,\boundeddistover$}
\label{sec:computation_of_boundeddist_boundeddistover}

With the sufficient conditions established for $\boundeddist,\boundeddistover$
in Theorem \ref{thm:subset-and-superset} we now focus our attention on methods
for computing bounded disturbance sets that satisfy these criteria. To
reiterate, our disturbance in \eqref{eq:nonlin}, $w_{k}$ is an assumed
i.i.d. disturbance drawn from the probability space $(\distspace,
\sigma(\distspace), \prob_{w_{k}})$. As mentioned in Remark
\ref{remark:non-unique-bounded-sets}, these bounded disturbance sets need not
be unique. Thus, there are many different methods that can be used to these
bounded sets. Here we propose an optimization problem to obtain generic
polyhedra that can be used to represent the bounded sets.

First, we define a polytope
\def\polytope#1#2{\hbox{Poly}(#1, #2)}
\begin{equation}
    \polytope{A}{b} = \big\{ y \in \realnums^{p}: A y \preceq b\}
\end{equation}
where $A \in \realnums^{q \times p}$, $b \in \realnums^{q}$, and $y \preceq
b$ if $y_{i} \leq b_{i}$ for all $i \in \{1, \dots, p\}$.
For a fixed polytopic shape, i.e. a fixed $A$, we formulate the optimization
problem,
\begin{mini!}
    {b}{\log(\mathrm{vol}(\polytope{A}{b})) \label{eq:choose_boundeddist_cost}}{\label{prob:choose_boundeddist}}{}
    % \addConstraint{0}{<a_{i}\label{eq:choose_{b}oundeddist_{c}onstraint1}}
    \addConstraint{\gamma^\frac{1}{N}}{\leq\prob_{w}(w_{k}\in \polytope{A}{b})\label{eq:choose_boundeddist_constraint2}}
\end{mini!}
To obtain $\boundeddist$ we use $\gamma = \alpha$, and for $\boundeddistover$,
$\gamma = 1 - \alpha$. The volume $\mathrm{vol}$ is the Lebesgue measure of the
set, 
\begin{equation*}
    \mathrm{vol}(A) = \int_{A} dx.
\end{equation*}
\begin{lemma}
    For $\theta \in [0,1]$, and $A \in \realnums^{q \times p}$ such that
    $\polytope{A}{0} = \{ 0 \}$, then $\polytope{A}{\theta b} = \theta
    \polytope{A}{b}$.
    \label{lem:scaled-polytopes}
\end{lemma}
\begin{proof}
    For $\theta = 0$, $\polytope{A}{0} = 0 \times \polytope{A}{b} = \{ 0 \}.$

    For $\theta \in (0,1]$, if $y \in \polytope{A}{\theta b}$ then $Ay \leq
    \theta b \Rightarrow A y/\theta \leq b \Rightarrow y \in \theta
    \polytope{A}{b}$. Thus $\polytope{A}{\theta b} \subseteq \theta
    \polytope{A}{b}$.

    If $y \in \polytope{A}{b}$, then $Ay \leq b \Rightarrow A \theta y \leq
    \theta b$. Let $z = \theta y$, then $z \in \polytope{A}{\theta b}
    \Rightarrow \theta \polytope{A}{b} \subseteq \polytope{A}{\theta b}$.
\end{proof}
\begin{proposition}
    For $A \in \realnums^{q \times p}$ such that $\polytope{A}{0} = \{ 0 \}$,
    if $\prob_{w}$ is a log-concave probability measure, then
    \eqref{prob:choose_boundeddist} is a concave minimization problem.
    \label{prop:log-concave-polytope}
\end{proposition}
\begin{proof}
    \def\convexpoly{\polytope{A}{\theta b_{1} + (1 - \theta) b_{2}}}
    \def\thetapoly{\polytope{A}{\theta b_{1}}}
    \def\oneminusthetapoly{\polytope{A}{(1 - \theta) b_{2}}}
    From~\cite[Ex. 3.44]{BoydConvex2004}, $\mathrm{vol}(\polytope{A}{b})$ is
    log-concave in $b\in \realnums^p$, implying the objective
    \eqref{eq:choose_boundeddist_cost} is concave.

    Next we show that, for $\theta \in [0, 1]$, and $b_{1}, b_{2} \in
    \realnums^{q}$,
    $$
        \thetapoly \minkadd \oneminusthetapoly \subset \convexpoly.
    $$ 
    Let $z_{1} \in \thetapoly$ and $z_{2} \in \oneminusthetapoly$. Let $y =
    z_{1} + z_{2}$, hence $y \in \thetapoly \minkadd \oneminusthetapoly
    \Rightarrow Ay = Az_{1} + Az_{2} \leq \theta b_{1} + (1 - \theta) b_{2}
    \Rightarrow y \in \convexpoly$.

    Now,
    \begin{align}
        &\prob_{w_{k}}\big(w_{k} \in \convexpoly \big) \nonumber \\
        &\hskip1cm \geq \prob_{w_{k}} \big( w_{k} \in \thetapoly \minkadd \oneminusthetapoly \big) \nonumber \\
        &\hskip1cm = \prob_{w_{k}} \big( w_{k} \in \theta \polytope{A}{b_{1}} \minkadd (1 - \theta) \polytope{A}{b_{2}} \big) \label{eq:logcon_{t}hetaoutside_{p}rob} \\
        &\hskip1cm \geq \prob_{w_{k}} \big( w_{k} \in \polytope{A}{b_{1}} \big)^{\theta} \prob_{w_{k}} \big( w_{k} \in \polytope{A}{b_{2}} \big)^{1 - \theta} \label{eq:logcon_{p}roduct}.
    \end{align}
    Equation \eqref{eq:logcon_{t}hetaoutside_{p}rob} follows from Lemma
    \ref{lem:scaled-polytopes} and \eqref{eq:logcon_{p}roduct} follows from the
    log-concavity of $\prob_{w_{k}}$. Since log-concavity implies
    quasiconcavity, the constraint \eqref{eq:choose_boundeddist_constraint2} is
    convex. Thus, the problem \eqref{prob:choose_boundeddist} minimizes a
    concave function over a convex set.

\end{proof}
\def\polyabc#1#2#3{\hbox{Poly}(#1, #2, #3)}
\begin{corollary}
    For $A \in \realnums^{q \times p}$, and polytope
    \begin{equation}
        \polyabc{A}{b}{c} = \big\{y \in \realnums^{p}: A(y - c) \preceq b \big\}
    \end{equation}
    where $\polyabc{A}{0}{c} = \{c\}$, if $\prob_{w}$ is a log-concave
    probability measure, then \eqref{prob:choose_boundeddist} is a concave
    minimization problem.
\end{corollary}
\begin{proof}
    Note that $\polyabc{A}{b}{c} = \polytope{A}{b} \minkadd \{c\}$. Let $v_{k}
    = w_{k} - c$, and the proof of Proposition \ref{prop:log-concave-polytope}
    holds for the random variable $v_{k}$.
\end{proof}

Multiplicative optimization problems belong to the class of concave
minimization (reverse convex optimization) problems, and they have been
well-studied in global optimization literature~\cite{benson1997multiplicative,
HorstGlobal2000}. They may be solved to global optimality using
branch-and-bound techniques. However, we employ a computationally simple
bisection method, described in Algorithm \ref{alg:bounded-set-bisection}, to
solve this problem to a potentially suboptimal solution~\cite[Sec.
3.3.3]{HorstGlobal2000}.

The use of generic polytopes defined by $A$, $b$, and $c$ allow for flexibility
in the definition of the bounded sets. For example, if $w_{k} \sim N(\mu,
\Sigma)$ then we can use Algorithm \ref{alg:bounded-set-bisection} to obtain a
cuboid bounded set by setting $A = \transpose{[\Sigma^{-1}, -\Sigma^{-1}]}$, $c
= \mu$, $b = \transpose{[1, \dots, 1]}$. If $w_{k}$ was drawn from an
exponential distribution we can obtain a cuboid bounded set with $A = I$, $c =
\transpose{[0, \dots, 0]}$, $b = \transpose{[1, \dots, 1]}$.
\begin{algorithm}
    \SetKwInOut{Input}{Input}
    \SetKwInOut{Output}{Output}

    \Input{Matrix, $A \in \realnums^{q \times p}$, and vectors $b \in \realnums^{q}, c \in \realnums^{p}$ defining initial Polyhedron $\polyabc{A}{b}{c}$, and $\gamma \in [0,1]$}
    \Output{Polyhedral bounded disturbance set $\mathcal{S}$}

    $\alpha \gets 1$ \\
    $p \gets \prob_{w_{k}}\big(w_{k} \in \polyabc{A}{\alpha b}{c}\big)$ \\

    \vskip 6pt
    // Expansion of polyhedron to find initial values for bisection \\
    \While{$p < \gamma^{\frac{1}{N}}$}{
        $\alpha \gets 2 \alpha$ \\
        $p \gets \prob_{w_{k}}\big(w_{k} \in \polyabc{A}{\alpha b}{c}\big)$ \\
    }

   \vskip6pt
    // Bisection \\
    \eIf{$\alpha = 1$}{
        $l \gets 0$
    }{
        $l \gets \alpha/2$ \\
    }
    $h \gets \alpha$ \\
    $m \gets \frac{l+h}{2}$ \\
    $p \gets \prob_{w_{k}}\big(w_{k} \in \polyabc{A}{m b}{c}\big)$ \\
    \While{$p \neq \gamma^{\frac{1}{N}}$}{
        \eIf{$p > \gamma^{\frac{1}{N}}$}{
            $h \gets m$ \\
        }{
            $l \gets m$ \\
        }

        $m \gets \frac{l+h}{2}$ \\
        $p \gets \prob_{w_{k}}\big(w_{k} \in \polyabc{A}{m b}{c}\big)$ \\
    }

    $\mathcal{S} \gets \polyabc{A}{m b}{c}$
    \caption{Bisection algorithm for obtaining polyhedral bounded disturbance sets. If $\gamma = \alpha$, then $\mathcal{S} = \boundeddist$, and if $\gamma = 1 - \alpha$, then $\mathcal{S} = \boundeddistover$.}
    \label{alg:bounded-set-bisection}
\end{algorithm}

\subsection{I.i.d. Gaussian disturbances}
\label{sub:iid-gaussian-disturbances}

For i.i.d. Gaussian disturbances we can determine minimum volume set satisfying
\eqref{eq:choose_boundeddist_constraint2} in the form of an ellipsoid. If the
disturbance $w_{k}$ in \eqref{eq:nonlin} is an $n$-dimensional Gaussian
random variable with mean vector $\mu \in \realnums^{n}$ and covariance matrix
$\Sigma \in \realnums^{n \times n}$, then its probability density function
is~\cite[Ch. 29]{billingsley_probability_1995}
\begin{align}
    \psi_{w_{k}}(y)={(2\pi)}^\frac{-n}{2}{\vert\Sigma\vert}^\frac{-1}{2}\exp{\left(-\frac{{(y-\mu)}^\top\Sigma^{-1}{(y-\mu)}}{2}\right)},
    \nonumber
\end{align}
$y \in \realnums^{n}$.

\def\rellipse{\mathcal{Q}_{R^{2}}}
Consider an $n$-dimensional ellipsoid, parameterized by $R^2\in[0,\infty)$,
\begin{align}
    \rellipse=\left\{y\in \realnums^{n}: {(y-\mu)}^\top\Sigma^{-1}{(y-\mu)}\leq
R^2\right\}.\label{eq:ellipsoid}
\end{align}
For $\mu=0,\Sigma=r^2 I_{n}$, we have $\rellipse=\{y: \transpose{y} y \leq r^2
R^2 \}$, a $n$-dimensional hypersphere of radius $rR$. We aim to compute the
parameter $R^2$ such that $\prob_{w_{k}}(w_{k}\in \rellipse\} =
\gamma^\frac{1}{N}$. If $\gamma = \alpha$, then $\boundeddist = \rellipse$ and
if $\gamma = (1 - \alpha)$ then $\boundeddistover = \rellipse$ will generate
bounded disturbance sets that satisfy the conditions of
Theorem~\ref{thm:subset-and-superset}.

\def\retaellipse{\mathcal{Q}_{\eta, R^{2}}}
Given a normally distributed $n$-dimensional random vector $\eta\sim
\mathcal{N}(0,I_{n})$, we have $w_{k} = {\Sigma}^\frac{1}{2}\eta+\mu$~\cite[Ch.
29]{billingsley_probability_1995}.  Also, $\rellipse={\Sigma}^\frac{1}{2}
\retaellipse\oplus\{\mu\}$ with $\retaellipse=\{s\in \realnums^{n}:s^\top
s\leq R^2\}$. Since the affine transformation of $\eta$ to $w_{k}$ is
deterministic, $\mathbb{P}_{w_{k}}(w_{k}\in\rellipse) =
\mathbb{P}_{\eta}(\eta\in \retaellipse) = \gamma^\frac{1}{N}$.  From~\cite[Ex.
20.16]{billingsley_probability_1995}, we have
$$
    F_{\chi^2(n)}(R^2)=\mathbb{P}\left\{\chi^2(n)\leq R^2\right\}=\mathbb{P}\{\eta\in \retaellipse\}=\gamma^\frac{1}{N}
$$
where $\chi^2(n)$ is a chi-squared random variable with $n$ degrees of freedom
and $F_{\chi^2(n)}(\cdot)$ denotes its cumulative distribution function.
Consequently, we have
\begin{align}
    R^2=F^{-1}_{\chi^2(n)}\left(\gamma^\frac{1}{N}\right).\label{eq:R2}
\end{align}
By solving \eqref{eq:R2} with $\gamma = \alpha$ or $\gamma = 1 - \alpha$ and then
using the result in \eqref{eq:ellipsoid}, we can obtain a feasible
$\boundeddist$ or $\boundeddistover$, respectively, for any Gaussian
disturbance.

% section computation_{o}f_{b}oundeddist_{b}oundeddistover_ (end)

\section{Computational challenges}
\label{sec:computational_challenges}

Since Lagrangian methods are grid-free, they have the potential to provide
substantial numerical benefits, most notably a dramatically increased
computational speed and applicability to high-dimensional systems. The
trade-off for the increased computational speed is a degree of conservativeness
in the approximations. It was shown in \cite{GleasonCDC2017} that for Gaussian
systems with low variance, these approximations become tight to the actual
solution obtained via dynamic programming. Still, there are several important
computational challenges at present that can limit the efficacy of these
methods for high-dimensional system analysis.

The primary challenge is the implementation of Algorithms \ref{alg:rets} and
\ref{alg:aets} with current computational geometry methods. In particular the
types of set operations required often limit the ability for current
methodologies to effectively scale with increasing dimension. For linear
systems, the one-step backward reach set it given by \eqref{eq:reachOlinCts},
thus for the disturbance maximal reach set we need intersections,
Minkowski addition, and affine transformations. To compute the
disturbance minimal reach set we need all the previous operations as well as
Minkowski (Pontryagin) difference and, potentially, unions, or 
convex hulls $\mathrm{Co}(\cdot)$ if multi-disturbance set methods
are used.

Table \ref{tab:set-operations-comp-geom} summarizes various capabilities for
current methodologies to perform the necessary set operations required to
compute the robust and augmented effective target sets. The following
subsections will provide additional detail regarding these different
computational geometry methods and examine their applicability for computing
the robust and augmented effective target sets.

\begin{table*}
    \definecolor{rowgray}{gray}{0.8}
    \def\tabunderapprox{\ensuremath{-}}
    \def\taboverapprox{\ensuremath{+}}
    \def\tabexact{\ensuremath{=}}
    \begin{center}
    \begin{tabular}{|c|c|c|c|c|c|c|}
        \hline
        Toolbox & \multicolumn{2}{c|}{{\tt MPT} \cite{MPT3}}      & {\tt ET} \cite{ellipsoidal} &                   & {\tt CORA} \\ \hline
        Methods & $\mathcal{H}$-polytope & $\mathcal{V}$-polytope & Ellipsoids                  & Support Functions & Zonotope   \\ \hline
        \rowcolor{rowgray} 
        $\cap$        & \tabexact &           & \tabunderapprox & \taboverapprox, \tabexact & \taboverapprox \\ \hline
        \rowcolor{rowgray}
        $T$           & \tabexact & \tabexact & \tabexact       & \tabexact                 & \tabexact      \\ \hline
        \rowcolor{rowgray}
        $\minkadd$    & \tabexact & \tabexact & \taboverapprox  & \tabexact                 & \tabexact      \\ \hline
        $\minkdiff$   & \tabexact &           & \tabunderapprox &                           &                \\ \hline
        $\cup$        &           & \tabexact & \taboverapprox  &                           &                \\ \hline
        $\mathrm{Co}$ &           & \tabexact & \taboverapprox  &                           &                \\ \hline
                      & \cite{guernic2009} & \cite{guernic2009} & \cite{kurzhanski2000-ellipsoid} & \cite{leguernic2010,guernic2009,maidens_2013} & \cite{guernic2009} \\ \hline
    \end{tabular}
    \end{center}
    % \vskip6pt
    \caption{
        Computational geometry tools and their ability to perform various set
        operations. To compute the the disturbance maximal reach set ass the
        operations shaded in gray are needed. For the disturbance minimal reach
        set, all six operations are needed (unions and convex hulls are only
        required when using multiple disturbances). Each method can compute
        the operation exactly ($\tabexact$) or will underapproximate 
        ($\tabunderapprox$) or overapproximate ($\taboverapprox$) the result. 
        References for the operations computability are provided at the base
        of the table.
        % Computational geometry methods ability to compute the set
        % operations required for the disturbance maximal reach set (shaded in gray) and 
        % disturbance minimal (all operations) effective target sets for linear systems. References for each
        % methods computational efficacy for the set operations are reported at the
        % end of its column. Each method shows an accuracy (Acc.) and difficulty
        % (Diff.). Accuracy can be either an underapproximation $\tabunderapprox$,
        % exact computation $\tabexact$, or and overapproximation $\taboverapprox$;
        % the difficulty is given as easy $1$, moderately difficult $2$, or hard $3$,
        % or not possible or unreported (blank cell).
    }
    \label{tab:set-operations-comp-geom}
\end{table*}

\subsection{Ellipsoidal Methods} % (fold)
\label{sub:ellipsoidal_methods}

Ellipsoidal methods \cite{kurzhanski2000-ellipsoid,ellipsoidal} can handle all
of the fundamental operations previously discussed but requires that all sets
be represented as ellipses. Hence for many methods---e.g. intersections,
unions, Minkowski addition---this methodology cannot compute exact
representations of the resulting set and thus determine approximations. For
example, the intersection of two ellipses does not necessary result in an
ellipse, thus ellipsoidal methods underapproximate these intersection by
finding the largest ellipse that exists inside of the intersection. Conversely
for unions a bounding ellipse is used to overapproximate the operation.

For the disturbance minimal and maximal reach sets these approximations make ellipsoidal
methods unappealing. For the maximal reach set, we need to
consistently overapproximate to maintain guarantees. However, ellipsoidal
methods underapproximate intersections which makes ellipsoidal methods not
applicable. Ellipsoidal
methods can handle the appropriate operations for the disturbance minimal reach
set, however because this computation already provides a conservative
underapproximation the additional underapproximations made from the ellipsoidal
methods compound the conservativeness. Hence ellipsoidal methods are not
advantageous.

% subsection ellipsoidal_{m}ethods (end)

\subsection{$\mathcal{H}$ and $\mathcal{V}$-polytopes} % (fold)
\label{sub:polytopes}

\def\hpoly{\mathcal{H}}
\def\vpoly{\mathcal{V}}

Perhaps the most applicable method, $\hpoly$ and $\vpoly$-polytopes represent
polytopic sets in either halfspace or vertex form, respectively. One of the
most current toolboxes for polytopic computational geometry is the Model 
Parametric Toolbox \cite{MPT3}, written
in MATLAB. For most set operations there are convenient ways to exactly compute
the aforementioned set operation using $\hpoly$ or $\vpoly$-polytope
representation. Many operations are only feasible for a specific representative
form, e.g. intersections with $\hpoly$-polytope form. Thus, to use these
methods to compute the effective target sets, we must be able to switch between
the two forms, which is limited by the {\em vertex-facet enumeration problem}.
This problem limits ability for these methods to extend to higher-dimensional
spaces as the vertex-facet problem becomes more and more costly.

Additionally polytopic representation also suffers from a need to approximation
generic shapes as polytopes. For example, representing the ellipsoid computed
in Section \ref{sub:iid-gaussian-disturbances}, $\rellipse$, can be problematic
In order to obtain an approximation of the ellipse using either $\hpoly$ or
$\vpoly$-polytope representations we need to to choose $m$ directions in $n$
dimensional space on the surface of the ellipse. For an underapproximation each
point becomes a vertex of the polytope and as an overapproximation the plane
tangential to the ellipsoid's surface at that point becomes a halfspace. Which
points on the surface would best approximate the ellipse is a hard problem but
a good heuristic is to use points that are equidistant, i.e. $x_{i}^{T}x_{j} =
c$ for all $i \neq j$, $c$ is some constant. Choosing equidistant vectors
eliminates the possibility of creating unbounded or very large
overapproximations of the ellipsoid from poor directional choice; instead
unboundedness or large overapproximations would arise if $m$ is very small,
especially compared to the dimension size $n$, e.g. if $m < n$. For a
two-dimensional system choosing the appropriate directions is a simple problem
in which $x_{i}^{T}x_{j} = \cos(2*pi/m)$, see Figure
\ref{fig:ellipse-approximatinos}.
\begin{figure}
  \centering
  \includegraphics{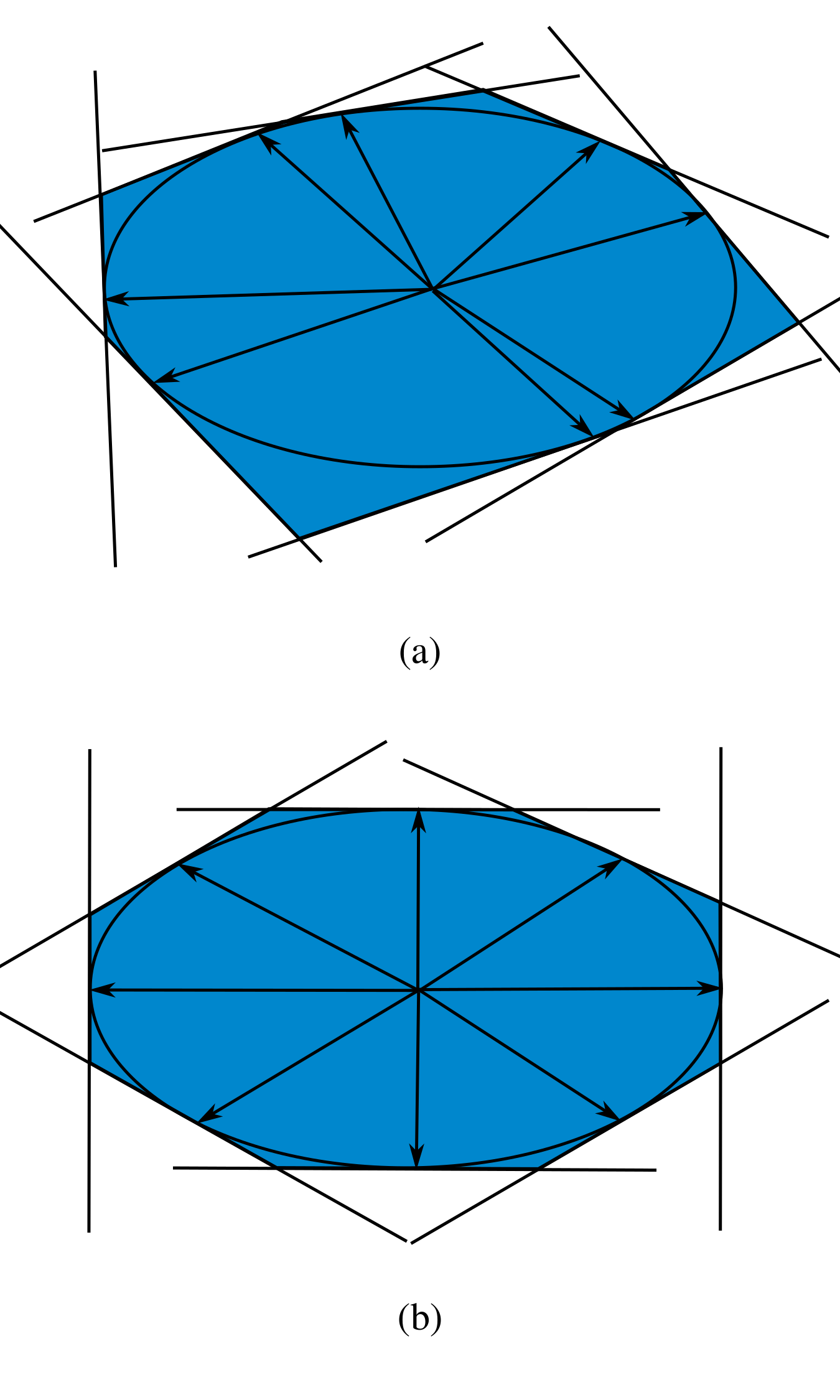}
  \caption{Ellipsoid approximations from (a) random directions (b) equidistant vector generation, $n = 2$, $m = 8$.}
  \label{fig:ellipse-approximatinos}
\end{figure}
In systems with $n>2$, however, approximations are done in a random manner
\cite{harman2010} by sampling points on the exterior of the ellipse, and hence
polyhedral representations will vary. 

Additionally, accurate approximations require an increasing number of
vertices/halfspaces as dimensionality increases, again increasing computational
requirements. The use of bounded disturbances that are constructed as simple
polytopic structures, as are made in Algorithm \ref{alg:bounded-set-bisection},
can help reduce computational difficulties in higher dimensions.
In \cite{gleason_2018} it was shown that using a single disturbance box allowed for the underapproximation methods to compute simulations on a 6-dimensional chain of integrators. Additionally for low-dimensional systems the average computation time for a disturbance box was 50\% of the computation time for a polytopic ellipsoid approximation.

Because the necessary set operations can be performed without additional
conservativeness through approximations, we employ polytopic methods and use
{\tt MPT} to compute the robust effective target set in the examples.

% subsection mathcal (end)

\subsection{Zonotopes} % (fold)
\label{sub:zonotopes}

Zonotopes have been demonstrated to be effective tools for higher-dimensional
reachability analysis \cite{althoff_2011}. For zonotopes,
intersections are represented as zonotope bundles, and the Minkowski sum of a set
and a zonotope bundle yields an \emph{overapproximation}. This
restricts us from using zonotopes for the disturbance minimal reach set since it
would nullify our conservativeness guarantee.

Because the disturbance maximal reach set is an overapproximation zonotopes are
a viable option for this computation. We do not use these methods in this paper,
however, because we do not wish to compound conservativeness, i.e. take 
overapproximations of an overapproximation (similar to ellipsoidal methods).

% subsection zonotopes (end)

\subsection{Star Methods} % (fold)
\label{sub:star_methods}

The star methods are capable of performing undisturbed reach computations 
\cite{bak2017} but have not been reported to be able to handle the set 
operations individually. Hence we cannot evaluate their efficacy with our
methods.

% [[We need to talk to HyLAA creators directly about this one, possibly, because
% from \cite{bak2017} it is very difficult to determine what of any of the set
% operations they can perform.]]

% subsection star_{m}ethods (end)

\subsection{Support Functions} % (fold)
\label{sub:support_functions}

Support functions are a powerful tool for representing generic convex sets
\cite{gardner2013,leguernic2010,maidens_2013,le_guernic_reachability_2009}. They have very simple and
effective methods for computing affine transformations and Minkowski summation
\cite{maidens_2013,guernic2009}, and can compute Minkowski differences with a
polytopic set \cite{kolmanovsky1998}. However, intersections with support
functions are more problematic. Exact computation of intersections can be
determined through a computationally difficult optimization problem
\cite{guernic2009} but they can be overapproximated simply. This
overapproximation eliminates their applicability for computation of the robust
effective target set. Additionally, this intersection approximation does not
permit additional Minkowski summation \cite{guernic2009}. 

% \subsection{Flowpipe Tools}

% subsection support_{f}unctions (end)

\section{Examples} % (fold)
\label{sec:examples}

All calculations were done using MATLAB R2017a on a computer with an Intel Xeon
E3-1270 v6 processor and 32 GB RAM (2400MHz DDR4 UDIMM ECC) running Ubuntu
16.04. The computations were performed using the Stochastic Reachability
Toolbox (SReachTools) which utilizes the model parametric toolbox 
(MPT 3.0) \cite{MPT3} for the polyhedra and set operations. All simulation 
code can be found at \texttt{https://hscl.unm.edu/software/code/}.

\subsection{Two-dimensional double integrator}
\label{sec:double-int-ex}

We first consider a simple double integrator example. This example allows for
a direct comparison of the conservatism of the results and a comparison of
the computations speed against dynamic programming methods. The underapproximation
methods were compared against dynamic programming in \cite{GleasonCDC2017}.
Here, we reiterate these results as well as demonstrate the overapproximation.
% We first consider a simple example which allows direct comparison of the
% proposed underapproximation as well as the approximative result via dynamic
% programming, for conservatism and computational speed.

The discretized double integrator dynamics are
\begin{equation}
  x_{k+1} = \left[ \begin{array}{cc}
    1 & T \\
    0 & 1
  \end{array}\right] x_{k} + \left[\begin{array}{c}
    \frac{T^{2}}{2} \\
    T
  \end{array}\right] u_{k} + w_{k}
  \label{eq:disys}
\end{equation}
with state $x_{k} \in \statespace\subseteq\realnums^{2}$, input $u_{k} \in
\inputspace \subseteq \realnums$, $T = 0.25$, and Gaussian disturbance $w_{k}
\sim \mathcal{N}(0,0.005 \cdot I_{2})$. We consider the viability problem, or 
equivalently the target tube problem with $\targetset_{k} = \{ x_{k} 
\in \mathcal{X} : |x_{i}| \leq 1, i \in \{1, \dots, n\} \}$ for all
$k \in \{0, \dots, N \}$. In this example $\boundeddist$ and $\boundeddistover$
are ellipsoidal sets obtained via \eqref{eq:R2} and \eqref{eq:ellipsoid} with
$\alpha = 0.8$.

Figure \ref{fig:highvar_doubleint} compares the underapproximation and 
overapproximation, via Algorithms \ref{alg:rets} and \ref{alg:aets}, 
to the level sets computed via dynamic programming, as
in \cite{summers2010_verification}. 
% The underapproximation is tighter initially
% and becomes progressively more conservative as the time horizon $N$ increases.
% For time horizons $N_{1}, N_{2} \in \mathbb{N}$ with $N_{2} > N_{1}$, we know
% that $\alpha^{\frac{1}{N_{2}}} \geq \alpha^{\frac{1}{N_{1}}}$ for all
% $\alpha\in[0,1]$. Hence, from Section \ref{sub:bound-dist-comp}, $R^{2}_{N_{2}}
% \geq R^{2}_{N_{1}}$, indicating that $\boundeddist_{N_{1}} \subseteq 
% \boundeddist_{N_{2}}$, implying increased conservativeness by
% \eqref{eq:retset}. For the example shown in Figure \ref{fig:double-int} with $N
% = 1, 2, 3, 4, 5$, $R^{2} = 3.22, 4.50, 5.27, 5.83, 6.26$, respectively.
A comparison between the total computation time for both approaches is provided
in Table \ref{tab:di-simtimes}. The accuracy of dynamic programming relies on
its grid size, resulting in a trade-off between accuracy and computation speed,
from which Algorithms \ref{alg:rets} and \ref{alg:aets} do not suffer. 

\begin{table}
  \begin{center}
    \begin{tabular}{|c|c|c|c|}
      \hline
      Grid Size & Dynamic Programming & Approximations & Ratio\\ \hline
      $41 \times 41$ & 8.16 & 0.98 & 8.3\\
      $82 \times 82$ & 59.76 & 0.98 & 60.9\\ \hline
    \end{tabular}
  \end{center}
  \caption{Computation times, in seconds, for the double integrator problem, 
           solved via dynamic programming and via Algorithm \ref{alg:rets} and 
           \ref{alg:aets}. The ratio of computation times is the computation 
           time via dynamic programming divided by computation time via 
           Lagrangian approximations.}
  \label{tab:di-simtimes}
\end{table}

Both the under and overapproximation in this example are conservative. This is
the result of the need to be robust, in the case of the underapproximation,
to all disturbances in a bounded set. For Gaussian disturbances, as the variance
of the disturbance reduces, this bounded set also decreases in size
as a direct consequence of \eqref{eq:ellipsoid}. We demonstrated previously \cite{GleasonCDC2017},
omitted here for brevity, that as the variance decreases the underapproximation
becomes tight to the solution obtained via dynamic programming. This result
holds true for the overapproximation as well.

% For systems with Gaussian disturbance processes that have a low variance, the
% underapproximation obtained through the Lagrangian methods tightly
% underapproximates the stochastic level set, and is computed significantly
% faster---over $7$ times faster than the dynamic programming approach using a
% $41 \times 41$ grid.  Figure~\ref{fig:dicomp-lowvar} shows a comparison of the
% stochastic level set and the Lagrangian underapproximation when the Gaussian
% disturbance is of the form $w_{k} \sim \mathcal{N}(0, 10^{-5} I_{2})$. The
% irregularities on the exterior of the stochastic level set are a numerical
% artifact from the state-space gridding.

\begin{figure*}
  \centering
  \includegraphics{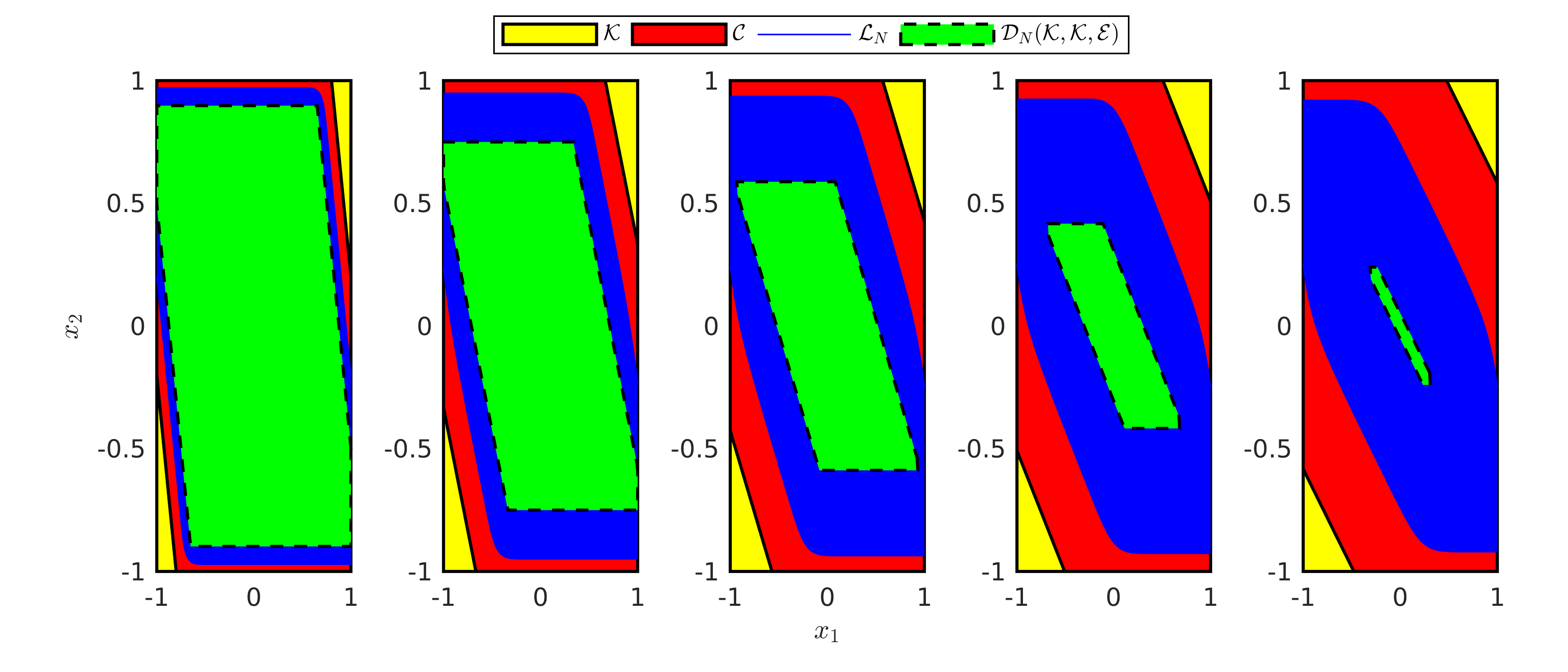}
  \caption{Double integrator example computation and comparison of 
           $\stochlevelset{k}$, $\reachunder{k}$ and $\reachover{k}$, where
           $\boundeddist$ and $\boundeddistover$ are ellipsoidal sets. Here,
           $N = 5$ and, from left to right, $k = 4, 3, 2, 1, 0$.}
  \label{fig:highvar_doubleint}
\end{figure*}

\subsection{Chain of integrators}

As was mentioned in \ref{sub:polytopes} because ellipsoidal representations of
polytopes can require high numbers of vertices or facets to represent as the
dimension of the system increases it is often advantageous, from a computational
time perspective, to use simpler polytopes that are obtained with Algorithm
\ref{alg:bounded-set-bisection}. Do demonstrate this we examine a chain of
integrators,
\begin{align}
   x_{k+1}&= \left[ {\begin{array}{ccccc}
   1 & T & \frac{1}{2}T^2 & \ldots  & \frac{1}{(n-1)!} T^{n-1}  \\  
   0 & 1 & T & & \\  
   \vdots & & & \ddots & \vdots \\
   0 & 0 & 0 & \ldots & T \\
   0 & 0 & 0 & \ldots & 1 \\
   \end{array} } \right]x_k \nonumber \\
   &\quad+ {\left[ {\begin{array}{cccc}
   \frac{1}{n!} T^n & \ldots & \frac{1}{2} T & T  \end{array} }
   \right]}^\top u_k + w_k\label{eq:integrator-chain}
\end{align}
with state $x_{k} \in \mathbb{R}^{n}$, input $u_{k} \in 
\mathcal{U} = \{ u \in \realnums : |u| \leq 0.1\}$, and
$w_{k} \sim \mathcal{N}(0, 1 \times 10^{-5} \cdot I_{n})$.

We analyze the same target tube (viability) problem as was done for the double
integrator with $T = 0.25$, $N = 5$, and $\alpha = 0.8$.
% We analyze the viability problem, in which we are interested in staying in a
% set of safe states for all $k = 0, \dots, N$. The safe set is
% \begin{equation}
%     \mathcal{K} = \{x \in \realnums^{2} : |x_{i}| \leq 1, i = 1, \dots, n \}.
%     \label{eq:integrator-safe-set}
% \end{equation}
With this example, we examine the scalability of these Lagrangian methods. 
For $n = 2, \dots, 6$, we computed 1) $\multreachunder{0}{\mathrm{box}}$
where $\boundeddist_{\mathrm{box}}$ is an origin-centered $n$-d box
\begin{equation}
    \boundeddist_{\hbox{box}} = \polyabc{\transpose{[I_{n}^{-1}, -I_{n}^{-1}]}}{b}{0}
\end{equation}
obtained via solution to Algorithm \ref{alg:bounded-set-bisection};
2) $\multreachunder{0}{R^{2}}$ where $\boundeddist_{R^{2}}$ is a polyhedral 
approximation of an ellipsoid given by \eqref{eq:ellipsoid}, sampled in 200 
random directions; and 3) an underapproximation, $\radlin{0}{[1,3]}$ 
where $\boundeddist_{1} = \boundeddist_{\hbox{box}}$ and
\begin{equation}
    \boundeddist_{i} = \polyabc{\transpose{[I_{n}^{-1}, -I_{n}^{-1}]}}{b}{c_{i}} \qquad i = 2, 3
\end{equation}
are off-center, $c_{i} \neq 0_{n}$, $n$-d boxes from Algorithm \ref{alg:bounded-set-bisection}. 
% The
% bounded sets $\boundeddist_{\hbox{box}}$ and $\boundeddist_{1}, 
% \boundeddist_{2}, \boundeddist_{3}$ were all obtained using Algorithm 
% \ref{alg:union-algorithm}. The set $\boundeddist_{\hbox{box}}$ and 
% $\boundeddist_{1}$ are cuboid polytopes $\polytope($ has with three bounded boxes.  

Figure \ref{fig:chain-of-int-comptimes} shows the computation times
for each set. We were unable to compute $\multreachunder{0}{R^{2}}$ and
$\radlin{0}{[1,3]}$ for $n > 4$.
\begin{figure}
    \centering
    \includegraphics{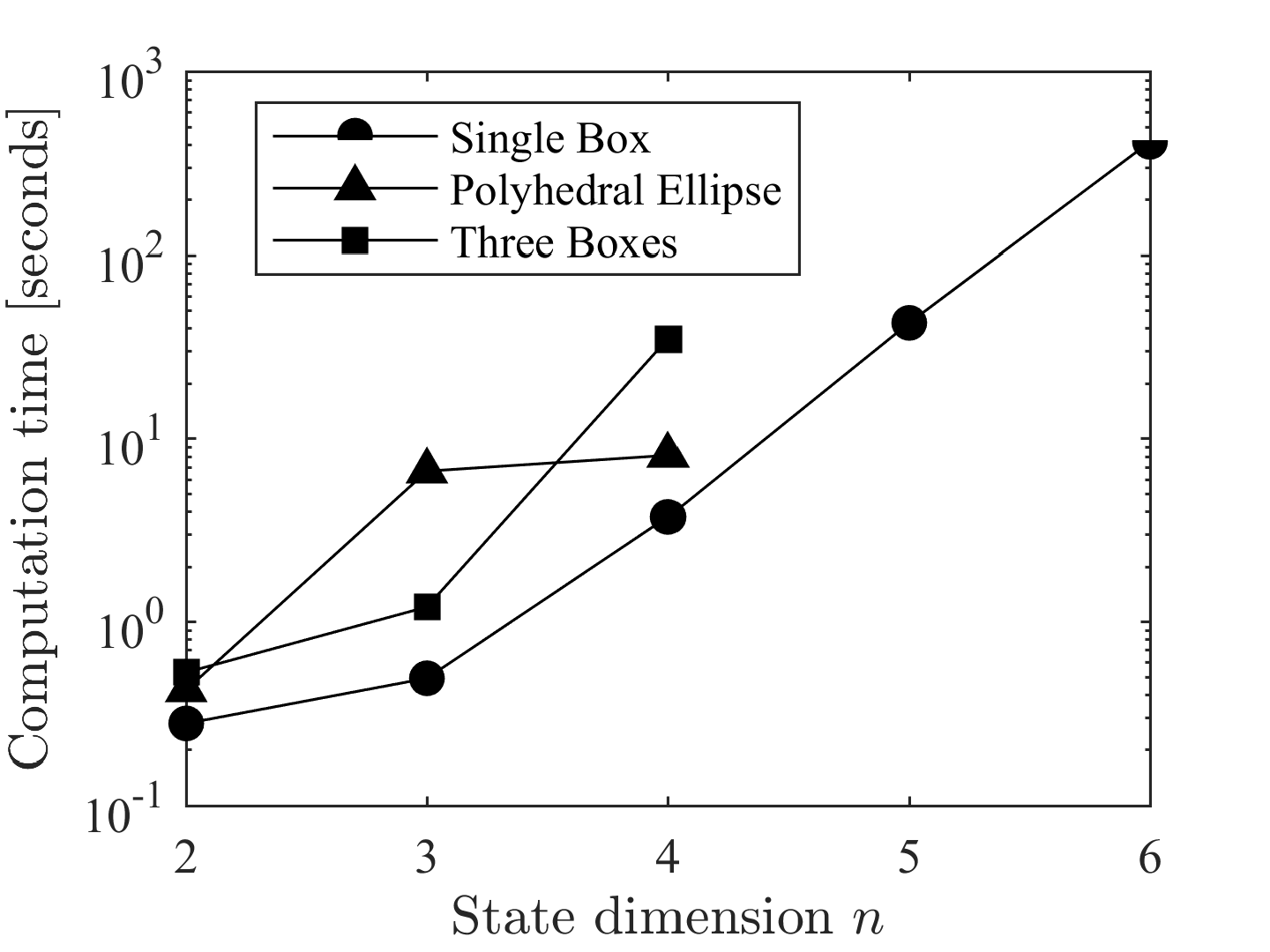}
    \caption{Computation time for $\multreachunder{0}{R^{2}}$ computed with a single polyhedral 
    overapproximation of an ellipse using 200 random directions (triangles),
    $\multreachunder{0}{\mathrm{box}}$ computed with a single box centered at the origin (circles), 
    and $\radlin{0}{[1,3]}$ with three boxes (squares).}
    \label{fig:chain-of-int-comptimes}
\end{figure}
For 2 and 3-dimensional integrators using multiple boxes was still faster than
computation using a polyhedral approximation an ellipse, and while disturbance 
minimal reach set could not be computed for the ellipsoid disturbance set, we
were able to simulate up to a 6-dimensional integrator when using a single $n$-d
box disturbance. The rise in computation time is predominately caused by the 
burden of solving the vertex-facet enumeration problem for higher dimensions.

\subsection{Application to space-vehicle dynamics}

We now consider a more realistic problem, motivated by the rendezvous and
docking problem for a pair of space vehicles. The goal is for one spacecraft,
referred to as the deputy, to approach and dock to an orbiting satellite,
referred to as the chief, while remaining in a predefined line-of-sight cone,
in which accurate sensing of the other vehicle is possible. The dynamics are
described by the Clohessy-Wiltshire-Hill (CWH) equations
\cite{wiesel1989_spaceflight}
\begin{equation}
  \begin{split}
    \ddot{x} - 3 \omega x - 2 \omega \dot{y} = \frac{F_{x}}{m_{d}} \\
    \ddot{y} + 2 \omega \dot{x} = \frac{F_{y}}{m_{d}}
  \end{split}
  \label{eq:2d-cwh}
\end{equation}
The chief is located at the origin, the position of the deputy is $x,y \in
\realnums$, $\omega = \sqrt{\mu/R_{0}^{3}}$ is the orbital frequency, $\mu$ is
the gravitational constant, and $R_{0}$ is the orbital radius of the spacecraft.

We define the state as $z = [x,y,\dot{x},\dot{y}] \in \realnums^{4}$ and input
as $u = [F_{x},F_{y}] \in \inputspace\subseteq\realnums^{2}$. We discretize
the dynamics (\ref{eq:2d-cwh}) in time to obtain the discrete-time LTI system,
\begin{equation}
  z_{k+1} = A z_{k} + B u_{k} + w_{k} \label{eq:lin-cwh}
\end{equation}
with $w_{k} \in \realnums^{4}$ a Gaussian i.i.d. disturbance, with
$\mathbb{E}[w_{k}] = 0$, $\Sigma = \mathbb{E}[w_{k}w_{k}^\top] =
10^{-4}\times\mbox{diag}(1, 1, 5 \times 10^{-4}, 5 \times 10^{-4})$.
The bounded disturbance sets $\boundeddist$ and $\boundeddistover$ are 
origin-centered 4-dimensional boxes obtained with Algorithm 
\ref{alg:bounded-set-bisection}.

We define the target sets and as in \cite{lesser2013_spacecraft}
\begin{align}
  \targetset_{N} &= \left\{ z \in \realnums^{4}: |z_{1}| \leq 0.1, -0.1 \leq z_{2} \leq 0, \right. \nonumber \\
              &\hskip1.9cm \left. |z_{3}| \leq 0.01, |z_{4}| \leq 0.01 \right\} \\
  \targetset_{i} &= \left\{z \in \realnums^{4}: |z_{1}| \leq z_{2}, |z_{3}| \leq 0.05, |z_{4}| \leq 0.05 \right\}\\
                 & \qquad i = 1, 2, \dots, N-1 \\
  \inputspace &= [-0.1, 0.1] \times [-0.1, 0.1].
\end{align}

For a horizon, $N = 5$, and a level set, $\alpha = 0.8$, $\boundeddist = \{ s :
s^{\top} \Sigma^{-1} s \leq 6.26\}$, from \eqref{eq:ellipsoid}. Figure
\ref{fig:cwh} shows a cross-section at $\dot{x} = \dot{y} = 0$ of the
resulting underapproximation of the $N=5$ stochastic reach-avoid level set. The
computation time for this level set was $14.5$ seconds. Direct comparison of
results via dynamic programming is not possible due to dimensionality of the
state. However, in \cite[Figure 2]{lesser2013_spacecraft}, a cross-section of
$\dot{x} = \dot{y} = 0.9$ of the stochastic reach-avoid set was approximated
via convex chance-constrained optimization and particle approximation methods. 
Although these approaches require gridding, they are computationally feasible,
unlike dynamic programming.  The computation time reported
in~\cite{lesser2013_spacecraft} is approximately $20$ minutes (about $82$ times
slower) for just a subset of the state space.
\begin{figure}
  \centering
  \includegraphics{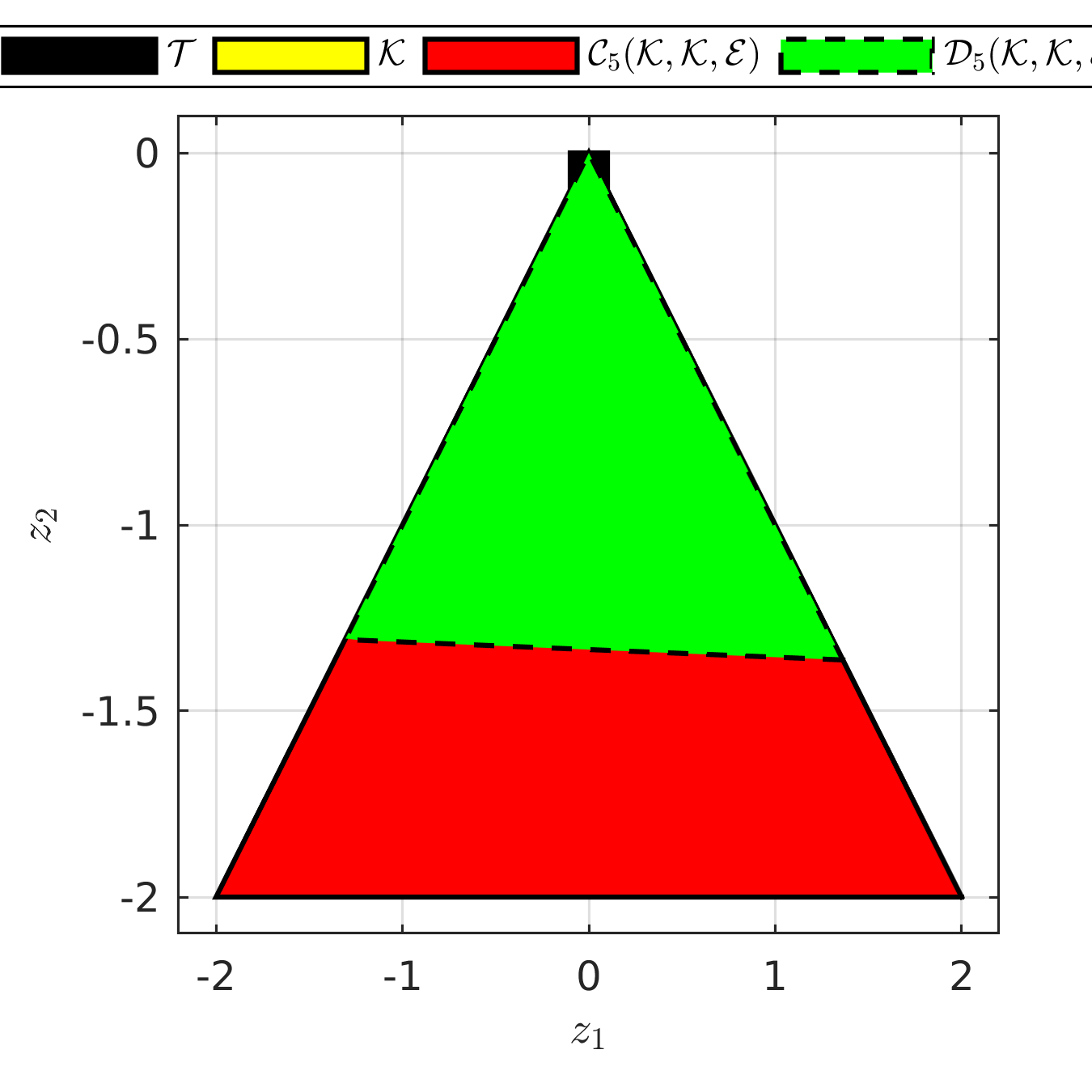}
  \caption{Computation of $\reachunder{0}$ and $\reachover{0}$ for the 
           4-dimensional CWH spacecraft rendezvous docking problem. Outside of
           the realm of dynamic programming we are still able to obtain an under
           and overapproximation of the stochastic reach set using Lagrangian
           methods.}
  \label{fig:cwh}
\end{figure}

\section{Conclusion} % (fold)
\label{sec:conclusion}

In this work we have demonstrated how Lagrangian recursions can be used to
under and overapproximate the stochastic $\alpha$-level reach set of the
reachability of the target tube problem. The reachability of a target tube
problem was shown to be a more generalized version of both the terminal-time
reach-avoid and viability problems. These recursive methods provide substantial
gains in computation time when compared to dynamic programming since they do
not require on gridding the system, nor are the accuracy of the obtained
results related to the spacing of the grid. The Lagrangian results are however,
conservative, and while we are able to demonstrate the ability to simulate 
systems with moderate dimensional size, the current limitations in computational
geometry tools limit this growth. In the examples demonstrated, the requirement
to solve the vertex-facet enumeration problem is a limiting factor.

In the future, we intend to continue work to allow for Lagrangian methods to 
provide approximations sets for higher dimensional systems. The use of zonotopes
for the overapproximation, for example, may allow for evaluation of higher
dimensional systems. Additionally, Lagrangian methods provide fast simulation
for a set of states for which existence of a closed-loop feedback controller
is guaranteed, however an explicit controller is not provided. We also plan to
develop methods that will use the disturbance minimal reach tube, i.e. 
$[\reachunder{0}, \reachunder{1}, \dots, \reachunder{N}]$ to determine 
open and closed-loop control strategies that will have the probabilistic 
guarantees of safety established by the approximation. Controllers can be 
equivalently found for the disturbance maximal reach tube.

% section conclusion (end)

\bibliographystyle{IEEEtran}
\bibliography{refs}

\end{document}

%% file: macrosv1.tex
\usepackage{xspace}
\usepackage{amsmath,amssymb,amsfonts}

\newcommand{\bs}[1]{\ensuremath{\boldsymbol{#1}}}

\newcommand{\pci}[1]{\ensuremath{\bs{p}_{\boldsymbol{c}_{i}}}\xspace}